\documentclass[11pt]{article}

\def\notes{0}
\ifnum\notes=1
\newcommand{\comment}[2]{\marginpar{\tiny{\textbf{#1: }\textit{#2}}}}
\else
\newcommand{\comment}[2]{}
\fi

\newcommand{\madhu}[1]{\comment{M}{#1}}
\newcommand{\srikanth}[1]{\comment{S}{#1}}
\newcommand{\mitali}[1]{\comment{Mitali}{#1}}

\usepackage{color}
\usepackage[pagebackref=true]{hyperref}
\hypersetup{
    unicode=false,          
    colorlinks=true,        
    linkcolor=blue,          
    citecolor=blue,        
    filecolor=magenta,      
    urlcolor=cyan           
}

\usepackage{hyperref,cleveref}

\usepackage{amsmath}
\usepackage{amssymb}
\usepackage{amsthm}
\usepackage{fullpage}
\usepackage{rotating}
\usepackage{tikz}
\usepackage{bbold}

\newtheorem{theorem}{Theorem}[section]
\newtheorem{corollary}[theorem]{Corollary}
\newtheorem{lemma}[theorem]{Lemma}
\newtheorem{observation}[theorem]{Observation}
\newtheorem{proposition}[theorem]{Proposition}
\newtheorem{definition}[theorem]{Definition}
\newtheorem{claim}[theorem]{Claim}
\newtheorem{fact}[theorem]{Fact}
\newtheorem{remark}[theorem]{Remark}

\newcommand{\prob}[2]{\mathop{\mathrm{Pr}}_{#1}\left[#2\right]}
\newcommand{\avg}[2]{\mathop{\textbf{E}}_{#1}\left[#2\right]}
\newcommand{\poly}{\mathop{\mathrm{poly}}}

\newcommand{\F}{\mathbb{F}}
\newcommand{\ip}[2]{\langle #1, #2 \rangle}

\newcommand{\Var}{\mathrm{Var}}

\newcommand{\mynorm}[1]{\left\lVert #1 \right\rVert}
\newcommand{\mc}[1]{\mathcal{#1}}
\newcommand{\mf}[1]{\mathfrak{#1}}

\newcommand{\Supp}{\mathrm{Supp}}
\newcommand{\charF}{\mathrm{char}}
\newcommand{\calf}{\mathcal{F}}
\newcommand{\Fnd}[2]{\mathcal{F}(#1,#2)}
\newcommand{\Fn}[1]{\mathcal{F}(#1)}

\usepackage{tikz}

\title{Local decoding and testing of polynomials over grids}
\author{Mitali Bafna\thanks{Harvard John A. Paulson School of Engineering and Applied Sciences. {\tt mitalibafna@g.harvard.edu}.} \and Srikanth Srinivasan\thanks{Department of Mathematics, IIT Bombay. {\tt srikanth@math.iitb.ac.in.}} \and Madhu Sudan\thanks{Harvard John A. Paulson School of Engineering and Applied Sciences. {\tt madhu@cs.harvard.edu}.}}

\begin{document}

\maketitle

\begin{abstract}
The well-known DeMillo-Lipton-Schwartz-Zippel lemma says that $n$-variate
polynomials of total degree at most $d$ over
grids, i.e. sets of the form $A_1 \times A_2 \times \cdots \times A_n$, form
error-correcting codes (of distance at least $2^{-d}$ provided $\min_i\{|A_i|\}\geq 2$).
In this work we explore their local
decodability and (tolerant) local testability. While these aspects have been studied
extensively when $A_1 = \cdots = A_n = \F_q$ are the same finite field, the
setting when $A_i$'s are not the full field does not seem to have been explored before.

In this work we focus on the case $A_i = \{0,1\}$ for every $i$. We show that for every field
(finite or otherwise) there is a test whose query complexity depends only on the
degree (and not on the number of variables). In contrast we show that
decodability is possible over fields of positive characteristic (with query complexity growing with
the degree of the polynomial and the characteristic), but not over the reals, where the query complexity must grow with $n$. As a
consequence we get a natural example of a code (one with a transitive group of symmetries) that is locally testable
but not locally decodable.

Classical results on local decoding and testing of polynomials have relied on the
2-transitive symmetries of the space of low-degree polynomials (under affine
transformations). Grids do not possess this symmetry: So we introduce some new
techniques to overcome this handicap and in particular use the hypercontractivity of
the (constant weight) noise operator on the Hamming cube.
\end{abstract}

\newpage

\tableofcontents

\newpage

\setcounter{page}{1}

\section{Introduction}

Low-degree polynomials have played a central role in computational complexity. (See for instance ~\cite{Shamir79,BCW,BGW,MVV,Razborov,LFKN,Shamir90,AroraS,ALMSS} for some of the early applications.)
One of the key properties of low-degree $n$-variate polynomials underlying many of the applications is the ``DeMillo-Lipton-Schwartz-Zippel'' distance lemma~\cite{DeMilloL,Schwartz,Zippel} which upper bounds the number of zeroes that a non-zero low-degree polynomial may have over ``grids'', i.e., over domains of the form $A_1 \times \cdots \times A_n$. This turns the space of polynomials into an error-correcting code (first observed by Reed~\cite{Reed} and Muller~\cite{Muller}) and many applications are built around this class of codes. These applications have also motivated a rich collection of tools including polynomial time (global) decoding algorithms for these codes, and ``local decoding''~\cite{BeaverF,Lipton,BLR} and ``local testing''~\cite{RubinfeldS,AKKLR,KaufRon} procedures for these codes.

Somewhat strikingly though, many of these tools associated with these codes don't work (at least not immediately) for all grid-like domains, but work only for the specific case of the domain being the vector space $\F^n$ where $\F$ is the field over which the polynomial is defined and $\F$ is finite. The simplest example of such a gap in knowledge was the case of ``global decoding''. Here, given a function $f:\prod_{i=1}^n A_i \to \F$ as a truth-table, the goal is to find a nearby polynomial (up to half the distance of the underlying code) in time polynomial in $|\prod_i A_i|$. When the domain equals $\F^n$ then such algorithms date back to the 1950s. However the case of general $A_i$ remained open till 2016 when Kim and Kopparty~\cite{KimK} finally solved this problem.

In this paper we initiate the study of local decoding and testing algorithms for polynomials when the domain is not a vector space. For uniformity, we consider the case of polynomials over hypercubes (i.e., when $A_i = \{0,1\} \subseteq \F$ for every $i$). We describe the problems formally next and then describe our results.

\subsection{Distance, Local Decoding and Local Testing}

We start with some brief notation. For finite sets $A_1,\ldots,A_n \subseteq \F$ and functions $f,g:
A_1 \times \cdots A_n \to\F$, let the distance between $f$ and $g$, denoted 
$\delta(f,g)$ be the quantity $\Pr_{a} [f(a) \ne g(a)]$ where $a$ is drawn uniformly from $A_1 \times \cdots \times A_n$. We say $f$ is $\delta$-close to $g$ if $\delta(f,g)\leq \delta$, and $\delta$-far otherwise. For a family of functions $\calf \subseteq 
\{h:A_1\times\cdots\times A_n \to \F\}$, let $\delta(\calf) = \min_{f\ne g\in \calf}\{\delta(f,g)\}$. 

To set the context for some of the results on local decoding and testing, we first recall the distance property of polynomials. If $|A_i|\geq 2$ for every $i$, the polynomial distance lemma asserts that the distance between any two distinct degree $d$ polynomials is at least $2^{-d}$. 
Of particular interest is the fact that for fixed $d$ this distance is bounded away from $0$, independent of $n$ or $|\F|$ or the structure of the sets $A_i$. In turn this behavior effectively has led to ``local decoding'' and ``local testing'' algorithms with complexity depending only on $d$ --- we define these notions and elaborate on this sentence next.

Given a family of functions $\calf$ from the domain $A_1 \times \cdots \times A_n$ to $\F$, we say $\calf$ is {\em $(\delta,q)$-locally decodable} if there exists a probabilistic algorithm that, given $a \in A_1\times \cdots \times A_n$ and oracle access to a function $f:A_1 \times \cdots \times A_n\to \F$ that is $\delta$-close to some function $p \in \calf$, makes at most $q$ oracle queries to $f$ and outputs $p(a)$ with probability at least $3/4$. (The existence of a $(\delta,q)$-local decoder for $\calf$ in particular implies that $\delta(\calf) \geq 2\delta$.)
We say that $\calf$ is $(\delta,q)$-\emph{locally testable} if there exists a probabilistic algorithm that makes $q$ queries to an oracle for $f:A_1\times \cdots \times A_n \to \F$ and accepts with probability at least $3/4$ if $f \in \calf$ and rejects with probability at least $3/4$ if $f$ is $\delta$-far from every function in $\calf$. We also consider the {\em tolerant testing} problem where functions that are close to the family $\calf$ are required to be accepted with high probability. Specifically we say that $\calf$ is $(\delta_1,\delta_2,q)$-tolerantly-locally-testable if it has a $(\delta_2,q)$-local tester that satisfies the additional condition that if $\delta(f,\calf) \leq \delta_1$ then the tester accepts $f$ with probability at least $3/4$. 

When $A_1 = \cdots = A_n = \F$ (and so $\F$ is finite) it was shown by Kaufman and Ron~\cite{KaufRon} (with similar results in Jutla et al.~\cite{JPRZ})  that the family of $n$-variate degree $d$ polynomials over $\F$ is $(\delta,q)$-locally decodable and $(\delta,q)$-locally testable for some $\delta = \exp(-d)$ and $q = \exp(d)$. In particular both $q$ and $1/\delta$ are bounded for fixed $d$, independent of $n$ and $\F$. Indeed in both cases $\delta$ is lower bounded by a constant factor of $\delta(\Fnd{n}{d})$ and $q$ is upper bounded by a polynomial in the inverse of $\delta(\Fnd{n}{d})$ where $\Fnd{n}{d}$ denotes the family of degree $d$ $n$-variate polynomials over $\F$, seemingly suggesting that the testability and decodability may be consequences of the distance. If so does this phenomenon should extend to the case of other sets $A_i \ne \F$ - does it? We explore this question in this paper.

In what follows we say that the family of degree $d$ $n$-variate polynomials is locally decodable (resp.  (tolerantly) testable) if there is bounded $q = q(d)$ and positive $\delta = \delta(d)$ such that $\Fnd{n}{d}$ is $(\delta,q)$-locally decodable (resp. (tolerantly) testable) for every $n$. The specific question we address below is when are the family of degree $d$ $n$-variate polynomials locally decodable and (tolerantly) testable when the domain is $\{0,1\}^n$.
(We stress that the choice of $\{0,1\}^n$ as domain is partly for simplicity and is equivalent to the setting of $|A_i| = 2$ for all $i$. Working with domains of other (and varying) sizes would lead to quantitative changes and we do not consider that setting in this paper.)

\subsection{Main Results}

Our first result (Theorem~\ref{thm:imposs-large-char-intro}) shows that even the space of degree $1$ polynomials is {\em not locally decodable} over fields of zero characteristic or over fields of large characteristic. 
This statement already stresses the main difference between the vector space setting ( domain being $\F^n$) and the ``grid'' setting (domain = $\{0,1\}^n$). One key reason underlying this difference is that the domain $\F^n$ has a rich group of symmetries that preserve the space of degree $d$ polynomials, where the space of symmetries is much smaller when the domain is $\{0,1\}^n$. 
Specifically the space of degree $d$ polynomials over $\F^n$ is ``affine-invariant'' (invariant under all affine maps from $\F^n$ to $\F^n$). The richness of this group of symmetries is well-known  to lead to local decoding algorithms (see for instance~\cite{AKKLR}) and this explains the local decodability of $\Fnd{n}{d}$ over the domain $\F^n$. Of course the absence of this rich group of symmetries does not rule out local decodability --- and so some work has to be done to establish Theorem~\ref{thm:imposs-large-char-intro}. We give an overview of the proof in Section~\ref{ssec:overview} and then give the proof in Section~\ref{sec:imposs}.

Our second result (Theorem~\ref{thm:decoding-small-char-intro}) shows, in contrast, that the class of {\em degree $d$ polynomials over fields of small characteristic are locally decodable}. Specifically, we show that there is a $q = q(d,p)<\infty$ and  $\delta = \delta(d,p)>0$ such that $\Fnd{n}{d}$ over the domain $\{0,1\}^n$ over a (possibly infinite) field $\F$ of characteristic $p$ is $(\delta,q)$-locally decodable. This is perhaps the first local-decodability result for polynomials over infinite fields. A key technical ingredient that leads to this result, which may be of independent interest, is that when $n = 2p^t$ (twice a power of the characteristic of $\F$) and $g$ is a degree $d$ polynomial for $d < n/2$ then $g(0)$ can be determined from the value of $g$ on the ball on Hamming weight $n/2$ (see Lemma~\ref{lem:useful}). 
Again, we give an overview of the proof in Section~\ref{ssec:overview} and then give the actual proof in Section~\ref{sec:ldsmall}.

Our final, and main technical, result (Theorems~\ref{thm:test-main-intro}~and~\ref{thm:tolerant-Fq-intro}) shows somewhat surprisingly that $\Fnd{n}{d}$ is {\em always} (i.e., over all fields) {\em tolerantly locally testable}. This leads to perhaps the simplest natural example of a locally testable code that is not locally decodable. We remark there are of course many examples of such codes (see, for instance, the locally testable codes of Dinur~\cite{Dinur}) but these are results of careful constructions and in particular not very symmetric. On the other hand $\Fnd{n}{d}$ over $\{0,1\}^n$ does possess moderate symmetry and in particular the automorphism group is transitive. We remark that for both our positive results (Theorems~\ref{thm:decoding-small-char-intro}~and~\ref{thm:test-main-intro}), the algorithms themselves are not obvious and the analysis leads to further interesting questions. We elaborate on these in the next section.

\subsection{Overview of proofs}
\label{ssec:overview}

\paragraph{Impossibility of local decoding over fields of large characteristic.}
In Section~\ref{sec:imposs} we show that even the family of affine functions over $\{0,1\}^n$ is not locally decodable.
The main idea behind this construction and proof is to show that the value of a affine function $\ell:\{0,1\}^n \to \F$ at $1^n$ can not be determined from its values on any set $S$ if $|S|$ is small (specifically $|S| = o(\log n/\log\log n)$) and $S$ contains only ``balanced'' elements (i.e., $x \in S \Rightarrow |\sum_i x_i - (n/2)| = O(\sqrt{n})$.
Since the space of affine functions from $\{0,1\}^n$ to $\F$ forms a vector space, this in turn translates to showing that no set of up to $|S|$ balanced vectors contain the vector $1^n$ in their affine span (over $\F$) and we prove this in Lemma~\ref{lem:smallspan}. 

Going from Lemma~\ref{lem:smallspan} to Theorem~\ref{thm:imposs-large-char} is relatively standard in the case of finite fields. We show that if one picks a random linear function and simply erase its values on imbalanced inputs, this leads to only a small fraction of error, but its value  at $1^n$ is not decodable with $o(\log n/\log \log n)$ queries. (Indeed many of the ingredients go back to the work of \cite{BSHS}, who show that a canonical non-adaptive algorithm is effectively optimal for linear codes, though their results are stated in terms of local testing rather than local decoding.) In the case of infinite fields one has to be careful since one can not simply work with functions that are chosen uniformly at random. Instead we work with random linear functions with bounded coefficients. The bound on the coefficients leads to mild complications due to border effects that need care. In Section~\ref{ssec:imposs-proof} we show how to overcome these complications using a counting (or encoding) argument.

The technical heart of this part is thus the proof of Lemma~\ref{lem:smallspan} and we give some idea of this proof next. Suppose $S = \{x^1,\ldots,x^t\}$ contained $x^0 = 1^n$ in its affine span and suppose $|\sum_{j=1}^n x^i_j - (n/2)| \leq n/s$ for all $i$.
Let $a_1,\ldots,a_t \in \F$ be coefficients such that $x^0 = \sum_{i} a_i x^i$ with $\sum_i a_i = 1$. Our proof involves reasoning about the size of the coefficients $a_1,\ldots,a_t$. To get some intuition why this may help, note that 
$$\frac{n}{2} = \left|\sum_{j=1}^n x^0_j -\frac{n}{2}\right| = \left|\sum_{i=1}^t a_i\cdot \left( \sum_{j=1}^n x^i_j - \frac{n}{2}\right)\right| 
\leq \sum_{i=1}^t |a_i|\cdot \left| \sum_{j=1}^n x^i_j -\frac{n}{2}\right| \leq \frac{n}{s}\cdot \sum_j |a_j|.$$
So in particular if the $a_j$'s are small, specifically if $|a_j|\leq 1$ then we conclude $t  = \Omega(s)$. But what happens if large $a_j$'s are used? To understand this, we first show that the coefficients need not be too large (as a function of $t$) - see Lemma~\ref{lem:Cramer}, and then use this to prove Lemma~\ref{lem:smallspan}. The details are in Section~\ref{ssec:imposs-tech}.

\paragraph{Local decodability over fields of small characteristic.}
The classical method to obtain a $q$-query local decoder is to find, given a target point $x^0 \in \F^n$, a distribution on queries 
$x^1,\ldots,x^q \in \F^n$  such that (1) $P(x^0)$ is determined by $P(x^1),\ldots,P(x^q)$ for every degree $d$ polynomial $P$, and (2) the query $x^i$ is independent of $x^0$ (so that an oracle $f$ that usually equals $P$ will satisfy $P(x^i) = f(x^i)$ for all $i$, with probability at least $3/4$). 
Classical reductions used the ``2-transitivity'' of the underlying space of automorphisms to guarantee that $x^i$ is independent of $x^j$ for every pair $i\ne j \in \{0,\ldots,q\}$ --- a stronger property than required! Unfortunately, our automorphism space is not ``2-transitive'' but it turns out we can still find a distribution that satisfies the minimal needs. 

Specifically, in our reduction we identify a parameter $k = k(p,d)$ and map each variable $x_\ell$ to either $y_j$ or $1-y_j$ for some $j=j(\ell) \in [k]$. This reduces the $n$-variate decoding task with oracle access to $f(x_1,\ldots,x_k)$ to a $k$-variate decoding task with access to the function $g(y_1,\ldots,y_k)$. Since there are only $2^k$ distinct inputs to $g$, decoding can solved with at most $2^k$ queries (if it can be solved at all). The choice of whether $x_\ell$ is mapped to $y_j$ or $1-y_j$ is determined by $x^0_j$ so that $f(x^0) = g(0^k)$.  Thus given $x^0$, the only randomness is in the choice of $j(\ell)$. We 
choose $j(\ell)$ uniformly and independently from $[k]$ for each $\ell$. For $y \in \{0,1\}^k$,  $x^y$ denote the corresponding query in $\{0,1\}^n$ (i.e., $g(y) = f(x^y)$). Given our choices, $x^y$ is not independent of $x^0$ for every choice of $y$. Indeed if $y$ has Hamming weight $1$, then $x^y$ is very likely to have Hamming distance $\approx n/k$ from $x^0$ which is far from independent. However if $y \in \{0,1\}^k$ is a balanced vector with exactly $k/2$ $1$s (so in particular we will need $k$ to be even), then it turns out $x^y$ is indeed independent of $x^0$. So we query only those $x^y$ for which $y$ is balanced. But this leads to a new challenge: can $P(0^k)$ be determined from the values of $P(y)$ for balanced $y$s? It turns out that for a careful choice of $k$ (and this is where the small characteristic plays a role) the value of a degree $d$ polynomial at $0$ is indeed determined by its values on balanced inputs (see Lemma~\ref{lem:useful}) and this turns out to be sufficient to build a decoding algorithm over fields of small characteristic. Details may be found in Section~\ref{sec:ldsmall}.

\paragraph{Local testability over all fields.}

We now turn to the main technical result of the paper, namely the local testability of polynomials over grids. All previous analyses of local testability of polynomials with query complexity independent of the number of variables have relied on symmetry either implicitly or explicitly. (See for example~\cite{KaufS} for further elaboration.) Furthermore many also depend on the local decodability explicitly; and in our setting we seem to have insufficient symmetry and definitely no local decodability. This forces us to choose the test and analysis quite carefully. 

It turns out that among existing approaches to analyses of local tests, the one due to Bhattacharyya et al~\cite{BKSSZ} (henceforth BKSSZ) seems to make the least use of local decodability and our hope is to be able to simulate this analysis in our case --- but the question remains: ``which tester should we use?". This is a non-trivial question since the BKSSZ test is a natural one in a setting with sufficient symmetry; but their analysis relies crucially on the ability to view their  test as a sequence of restrictions: Given a function $f:\F^n\to\F$ they produce a sequence  of functions $f = f_n,f_{n-1},\ldots,f_k$, where the function $f_r$ is an $r$-variate function obtained by restricting $f_{r+1}$ to a codimension one affine subspace. Their test finally checks to see if $f_k$ is a degree $d$ polynomial. To emulate this analysis, we design a somewhat artificial test: We also produce a sequence of functions $f_n,f_{n-1},\ldots,f_k$ with $f_r$ being an $r$-variate function. Since  we do not have the luxury to restrict to arbitrary subspaces, we instead derive $f_{r}$ from $f_{r+1}(z_1,\ldots,z_{r+1})$ by setting  $z_i = z_j$ or $z_i = 1 - z_j$ for some random pair $i,j$ (since these are the only simple affine restrictions that preserve the domain). We stop when the number of variables $k$ is small enough (and hopefully a number depending on $d$ alone and not on $n$ or $\F$). We then test that the final function has  degree $d$.

The analysis of this test is not straightforward even given previous works, but we are able to adapt the analyses to our setting. Two new ingredients that  appear in our analyses are the hypercontractivity of hypercube with the constant weight noise operator (analyzed by Polyanskiy~\cite{Polyanskiy}) and the intriguing stochastics of a random set-union problem. We explain our analysis and where the above appear next.

We start with the part which is more immediate from the BKSSZ analysis. This corresponds to a key step in the BKSSZ analysis where it is shown that if $f_{r+1}$ is far from degree $d$ polynomials then, with high probability, so also is $f_r$. This step is argued via contradiction. If $f_r$ is close to the space of degree $d$ polynomials for many restrictions, then from the many polynomials that agree with $f_r$ (for many of the restrictions) one can glue together an $r+1$-variate polynomial that is close to $f_{r+1}$. This step is mostly algebraic and works out in our case also; though the actual algebra is different and involves more cases. (See Lemma~\ref{lem:LDL} and its proof in Section~\ref{sec:LDL}.) 

The new part in our analysis is in the case where $f_n$ is moderately close to some low-degree polynomial $P$. In this case we would still like to show that the test rejects $f_n$ with positive probability. In both BKSSZ and in our analysis this is shown by showing that the $2^k$ queries into $f_n$ (that given the entire truth table of the function $f_k$) satisfy the property that $f_n$ is not equal to $P$ on exactly one of the queried points. Note that the value of $f_k(y)$ is obtained by querying $f$ at some point, which we denote $x^y$. In the BKSSZ analysis $x^a$ and $x^b$ are completely independent given $a \ne b \in \{0,1\}^k$. (Note that the mapping from $y$ to $x^y$ is randomized and depends on the random choices of the tester.) In our setting the behavior of $x^a$ and $x^b$  is more complex and depends on both the set of coordinates $j$ such that where $a_j \ne b_j$ and on the number of indices $i\in [n]$ such that the variable $x_i$ is mapped to variable $y_j$.
Our analysis ends up depending on two new ingredients:  (1) The number of variables $x_i$ that map to any particular variable $y_j$ is $\Omega(n/k)$ with probability at least $2^{-O(k)}$ (see Corollary~\ref{cor:pdist}). This part involves the analysis of a random set-union process elaborated on below. (2) Once the exact number of indices $i$ such that $x_i$ maps to $y_j$ is fixed for every $j\in[k]$ and none of the sets is too small, the distribution of $x^a$ and $x^b$ is sufficiently independent to ensure that the events $f(x^a) = P(x^a)$ and $f(x^b) = P(x^b)$ co-occur with probability much smaller than the individual probabilities of these events. This part uses the hypercontractivity of the hypercube but under an unusual noise operator corresponding to the ``constant weight operator'', fortunately analyzed by Polyanskiy~\cite{Polyanskiy}. Invoking his theorem we are able to conclude the proof of this section.

We now briefly expand on the ``random set-union'' process alluded to above. Recall that our process starts with $n$ variables, and at each stage a pair of remaining variables is identified and given the same name. (We may ignore the complications due to the complementation of the form $z_i = 1-z_j$ for this part.) Equivalently we start with $n$ sets $X_1,\ldots,X_n$ with $X_i = \{i\}$ initially. We then pick two random sets and merge them. We stop when there are $k$ sets left and our goal is to understand the likelihood that one of the sets turn out to be too tiny. (The expected size of a set is $n/k$ and too tiny corresponds to being smaller than $n/(4k)$.)  It turns out that the distribution of set sizes produced by this process has a particularly clean description as follows: Randomly arrange the elements $1$ to $n$ on a cycle and consider the partition into $k$ sets generated by the set of elements that start with a special element and end before the next special element as we go clockwise around the cycle, where the elements in $\{1,\ldots,k\}$ are the special ones.  The sizes of these partitions are distributed identically to the  sizes of the sets $S_j$! For example, when $k=2$ the two sets have sizes distributed uniformly from $1$ to $n-1$. In particular the sets size are not strongly concentrated around $n/k$ - but nevertheless the probability that no set is tiny is not too small and this suffices for our analysis. 

Details of this analysis may be found in Section~\ref{sec:ltall}.

\paragraph{Tolerant Local Testability.}
\madhu{FILL THIS IN. Why is tester not tolerant? What modifications? What changes in analysis.}
We now give an outline of how one can augment the above local tester to get a tolerant local tester, i.e. for a fixed $\delta = \delta(d)$, given any $\delta_1 < \delta_2 < \delta$, the tolerant tester accepts all functions that are $\delta_1$-close with high probability and rejects functions that are $\delta_2$-far with high probability. We note that two routes that are often used in the context of algebraic properties to get tolerant testers are not available to us here. First, it is well-known that if a code is locally decodable and testable then it has a tolerant tester (since the local-decoder effectively gives oracle access to the nearby codeword). This route is not available to us since we don't have local decoder (for general fields). A second route that gives somewhat weak tolerant testers, say when $\delta_2 = 2\delta_1$, is to show that the rejection probability of $q$-query local tester is sandwiched between $\delta\cdot q \cdot (1 - (q-1)\delta)$ and $\delta \cdot q$, given a function at distance $\delta$ from the code. If such a statement were to be true, then the rejection probability of the test is a factor two approximation to the distance of the function from the code (provided $\delta$ is smaller than $1/(2q)$) and this leads to a weak tolerant tester. In our case our analysis only lower bound of the form $q\cdot \delta^{1.1}$ and this makes the approximation very weak. Thus getting a tolerant tester turns out to be non-trivial and below we explain how our tester works.

At a high level, our tolerant tester estimates the distance of $f$ to $\Fnd{n}{d}$. 
First, using the ``intolerant'' local tester, one can reject functions that are $\delta$-far from $\Fnd{n}{d}$, where $\delta$ is chosen to be $1/2^{O(d)}$. If the intolerant tester accepts, we produce a random restriction $R: \{0,1\}^k \rightarrow \{0,1\}^n$ on $k$ variables. This subspace sampling process is slightly different from the random restriction used in the local (intolerant) tester and we leave the details (which are not crucial at this stage) to Section~\ref{sec:tol-test}. We denote the restricted function produced by $g(y_1,\ldots,y_k)$, where $g(y) := f(R(y))$. Once we have produced a random subspace ($R(\{0,1\}^k)$), instead of querying $g$ on all points $\{0,1\}^k$, we further choose a random subset $S \subseteq \{0,1\}^k$ and find $g(y), \forall y \in S$. We use the queried values to do a brute-force interpolation to find the closest degree $d$ polynomial to $g$ on $S$ denoted by $h(y_1,\ldots,y_k)$. We then find $\delta_S(g,h)$, the distance between $g,h$ on $S$, and if this distance is small enough (depending on $\delta_1,\delta_2$) we accept, else we reject. 

Given that the intolerant tester accepts, let $p \in \Fnd{n}{d}$ denote the (uniquely) closest polynomial to the given function $f$. Normally one would query $g$ on all points of the random subspace (i.e. $g(\{0,1\}^k) = f(R(\{0,1\}^k))$) and check its distance to the closest low degree polynomial, but it turns out that the random set $S$ plays a crucial role in our analysis. The crux of the proof is in showing that the sampling process $R$ and the random set $S$ satisfy two main properties: (1) We show that the set of points $R(S)$ produced by the random process is a good sampler, that is, the fractional size of any subset $T \subseteq \{0,1\}^n$ in $R(S)$ is a good estimate of its fractional size in $\{0,1\}^n$. This holds because of the random properties of $S$ and the hypercontractivity of the noisy hypercube which is relevant due to the nature of $R$. Details can be found in Lemma~\ref{lem:poly-prop}. Note that such a strong sampling property does not hold if one considers all of $R(\{0,1\}^k)$ and so one has to consider $R(S)$.
(2)  In Corollary~\ref{cor:prop_S}, we show that the space of degree $d$ polynomials (over any field) have distance close to $1/2^d$ even when restricted to a random subset $S \subseteq \{0,1\}^k$, with high probability. (By Schwartz-Zippel we have that distance of degree $d$ polynomials on $\{0,1\}^k$ is $1/2^d$ but Theorem~\ref{cor:prop_S} shows that the same is true when restricted to a random subset $S \subseteq \{0,1\}^k$.) In fact it is interesting that an analogous property holds more generally for linear codes, as we show in Lemma~\ref{lem:random_S}.

These two properties can be combined in a straightforward way (See Section~\ref{sec:tol-test} for details) to give the correctness of the tolerant tester. Taking $T \subseteq \{0,1\}^n$ to be the set of points where $f$ differs from $p$, the first property shows that $\delta_{R(S)}(f,p) $ is a good estimate of the fractional distance $\delta(f,p)$ in $\{0,1\}^n$. By definition of the restriction, we have that $\delta_S(g,p\circ R) = \delta_{R(S)}(f,p)$. The second property shows that the interpolated polynomial $h$ is indeed equal to the true polynomial $p \circ R = p|_{R(S)}$. Hence $\delta_S(g,h) = \delta_{S}(g,p\circ R) = \delta_{R(S)}(f,p) \approx \delta_{\{0,1\}^n}(f,p)$, which shows that the tolerant tester is correct with high probability.

\paragraph{Organization.}
In Section~\ref{sec:prelims} we start with some preliminaries including the main definitions and some of the tools we will need later.
In Section~\ref{sec:results} we give a formal statement of our results.
In Section~\ref{sec:ltall} we present and analyze the local tester over all fields.
In Section~\ref{sec:imposs} we show that over fields of large (or zero) characteristic, local decoding is not possible.
In Section~\ref{sec:ldsmall} we give a local decoder and its analysis over fields of small characteristic.
Finally in Section~\ref{sec:tol-test} we present and analyse the tolerant tester over all fields.

\section{Preliminaries}
\label{sec:prelims}

\subsection{Basic notation} 

Fix a field $\F$ and an $n\in \mathbb{N}$. We consider functions $f:\{0,1\}^n \rightarrow \F$ that can be written as \emph{multilinear} polynomials of total degree at most $d$. We denote this space by $\Fnd{n}{d; \F}$. The space of all functions from $\{0,1\}^n$ to $\F$ will be denoted simply as $\Fn{n;\F}$. (We will simplify these to $\Fnd{n}{d}$ and $\Fn{n}$ respectively, if the field $\F$ is  clear from context.)

Given $f,g\in \Fn{n}$, we use $\delta(f,g)$ to denote the fractional Hamming distance between $f$ and $g$. I.e.,
\[
\delta(f,g) := \prob{x\in \{0,1\}^n}{f(x) \neq g(x)}
\]

For a family $\mc{F}'\subseteq \Fn{n}$, we use $\delta(f,\mc{F}')$ to denote $\min_{g\in \mc{F}'}\{\delta(f,g)\}.$ Given an $f\in \Fn{n}$ and $d\geq 0$, we use $\delta_d(f)$ to denote $\delta(f,\Fnd{n}{d}).$ 

\subsection{Local Testers and Decoders} 

Let $\mathbb{F}$ be any field. We define the notion of a local tester and local decoder for subspaces of $\Fn{n}$. 

\begin{definition}[Local tester]
Fix $q\in \mathbb{N}$ and $\delta\in (0,1)$. Let $\mc{F}'$ be any subspace of $\Fn{n}.$ 

We say that a randomized algorithm $T$ is a \emph{$(\delta,q)$-local tester for $\mc{F}'$} if on an input $f\in \Fn{n}$, the algorithm does the following.
\begin{itemize}
\item $T$ makes at most $q$ non-adaptive queries to $f$ and either accepts or rejects.
\item (Completeness) If $f\in \mc{F}'$, then $T$ accepts with probability at least $3/4$.
\item (Soundness) If $\delta(f,\mc{F}')\geq \delta$, then $T$ rejects with probability at least $3/4$.
\end{itemize}

Further, we say that $T$ is a \emph{$(\delta_1,\delta_2,q)$-tolerant local tester for $\mc{F}'$} if it is a $(\delta_2, q)$-local tester for $\mc{F'}$ which satisfies the following strengthened completeness property.
\begin{itemize}
\item (Tolerant Completeness) If $\delta(f,\mc{F}')\leq \delta_1$, then $T$ accepts with probability at least $3/4$.
\end{itemize}

We say that a tester is \emph{adaptive} if the queries it makes to the input $f$ depend on the answers to its earlier queries. Otherwise, we say that the tester is \emph{non-adaptive.}
\end{definition}

\begin{definition}[Local decoder]
Fix $q\in \mathbb{N}$ and $\delta\in (0,1)$. Let $\mc{F}'$ be any subspace of $\Fn{n}.$ 

We say that a randomized algorithm $T$ is a \emph{$(\delta,q)$-local decoder for $\mc{F}'$} if on an input $f\in \Fn{n}$ and $x\in \{0,1\}^n$, the algorithm does the following.
\begin{itemize}
\item $T$ makes at most $q$ queries to $f$ and outputs $b\in \F$.
\item If $\delta(f,\mc{F}')\leq \delta$, then the output $b = f(x)$ with probability at least $3/4$.
\end{itemize}
We say that a decoder is \emph{adaptive} if the queries it makes to the input $f$ depend on the answers to its earlier queries. Otherwise, we say that the tester is \emph{non-adaptive.}
\end{definition}

\subsection{Some basic facts about binomial coefficients}

\begin{fact}
\label{fac:Hamballs}
For integer parameters $0\leq b\leq a$, let $\binom{a}{b}$ denote the size of a Hamming ball of radius $b$ in $\{0,1\}^a$; equivalently, $\binom{a}{\leq b} = \sum_{j\leq b}\binom{a}{j}.$ Then, we have
\[
\binom{a}{\leq b}\leq 2^{aH(b/a)}
\]
where $H(\cdot)$ is the binary entropy function. 
\end{fact}

\subsection{Hypercontractivity.}

Hypercontractivity bounds $\ell_q$ norms of functions obtained by averaging a function of bounded $\ell_p$ norm, for some $p < q$, over weighted local neighborhoods. The traditional setting considers the averaging using the $\rho$-biased product measure over $\{0,1\}^n$, also known as the noisy hypercube (see Corollary~\ref{cor:noisy-hyp} below). In this work we also need to consider the setting where the neighborhood is the sphere of a specified distance - for which a hypercontractivity result was established recently by Polyanskiy~\cite{Polyanskiy} (see Theorem~\ref{thm:polanskiy} below).
Let us first recall the definition of $\rho$-biased product measure over $\{0,1\}^n$. 
\begin{definition}
	Let $\rho \in [0,1]$. For fixed $x \in \{0,1\}^n$, we write $y \sim N_\rho(x)$ to denote that the random string $y$ is drawn as follows: for each $i \in [n]$ independently,
	$$
	y_i = \begin{cases}
	x_i &\text{with probability } \rho \\
	\text{uniformly random } &\text{with probability } 1-\rho.
	\end{cases}
	$$ 
\end{definition}

\mitali{I didn't define $T_\rho$ and give the hypercontracivity theorem since we only require the SSE property which is stated directly in O'Donnell.}

For this product measure, Bonami~\cite{bonami} established a hypercontractivity theorem. One of the applications of the hypercontractivity theorem due to Bonami is the generalized small set expansion property of the noisy hypercube, details can be found in the book by O'Donnell~\cite[Chapter 10]{donnell}. 

\begin{theorem}\label{thm:gen-hyp}
	Let $0 \leq \rho \leq 1$. Let $A,B \subseteq \{-1,1\}^n$ have volumes $\exp\left(\frac{-a^2}{2} \right)$, $\exp\left(\frac{-b^2}{2}\right)$ and assume $0 \leq \rho a \leq b \leq a$. Then,
	$$\Pr_{\substack{ x \sim \{-1,1\}^n \\ y \sim N_\rho(x)}} [x \in A, y \in B] \leq  \exp\left(-\frac{1}{2}\frac{a^2 - 2\rho ab + b^2}{1 - \rho^2}\right ).$$
\end{theorem}

\begin{corollary}\label{cor:noisy-hyp}
	Let $A \subseteq \{-1,1\}^n$ have volume $\alpha$, i.e. let $\mathbb{1}_A:\{-1,1\}^n \rightarrow \{0,1\}$ satisfy $E[\mathbb{1}_A] = \alpha$. Then for any $\rho \in [-1,1]$, 
	$$\Pr_{\substack{ x \sim \{-1,1\}^n \\ y \sim N_\rho(x)}} [x \in A, y \in A] \leq  \alpha^{\frac{2}{1+|\rho|}}.$$
\end{corollary}

\begin{proof}
	We will prove the corollary in two steps, first for positive $\rho \in [0,1]$ and then for negative $\rho \in [-1,0]$. For $\rho \in [0,1]$, we will apply Theorem~\ref{thm:gen-hyp} with $B = A$ so that the volume of $B$, $\exp\left(\frac{-b^2}{2} \right)$ is equal to the volume of $A$, $\exp\left(\frac{-a^2}{2} \right)$ which in turn equals $\alpha$. We get that,
	$$\Pr_{\substack{ x \sim \{-1,1\}^n \\ y \sim N_\rho(x)}} [x \in A, y \in A] \leq  \exp\left(-\frac{1}{2}\frac{a^2 - 2\rho ab + b^2}{1 - \rho^2}\right ) = \exp\left(-\frac{1}{2}\frac{2a^2 - 2\rho a^2 }{1 - \rho^2}\right ) = \alpha^{\frac{2}{1+|\rho|}}.$$
	
	Now let us consider the case where $\rho \in [-1,0]$. Let $\overline{A} \subseteq \{-1,1\}^n$ denote the bitwise complement set of $A$, i.e $\overline{A} = \{a \mid \overline{a} \in A\}$, where the string $\overline{a}$ is the bitwise complement of $a$. Note that $a \in A \Leftrightarrow \overline{a} \in A$ which implies that $|A| = |\overline{A}|$. We also have that if $y \sim N_\rho(x)$ then $\overline{y} \sim N_{-\rho}(x)$. Applying Theorem~\ref{thm:gen-hyp} with $B = \overline{A}$ and $\exp\left(\frac{-b^2}{2} \right) = \exp\left(\frac{-a^2}{2} \right) = \alpha$ and $-\rho$ (which is positive) we get that,
	
	$$\Pr_{\substack{ x \sim \{-1,1\}^n \\ y \sim N_{\rho}(x)}} [x \in A, y \in A] = 
	\Pr_{\substack{ x \sim \{-1,1\}^n \\ \overline{y} \sim N_{-\rho}(x)}} [x \in A, \overline{y} \in \overline{A}]  
	 \leq  \exp\left(-\frac{1}{2}\frac{a^2 - 2(-\rho) ab + b^2}{1 - \rho^2}\right) = \alpha^{\frac{2}{1+|\rho|}}.$$
\end{proof}

Now we will get to the second setting, that of the spherical neighborhoods. Henceforth, let $\mathbb{R}$ be the underlying field. Let $\eta \in (0,1)$ be arbitrary. We define a smoothing operator $T_\eta$, which maps $\Fn{r} = \{f:\{0,1\}^r \rightarrow \mathbb{R}\}$ to itself. For $F\in \Fn{r}$, we define $T_\eta F$ as follows

\[
T_\eta F(x) = \avg{J\in \binom{[r]}{\eta r}}{ F(x\oplus J)}
\]
where $x\oplus J$ is the point $y\in \{0,1\}^r$ obtained by flipping $x$ at exactly the coordinates in $J$. 

Recall that for any $F\in \Fn{r}$ and any $p\geq 1$, $\mynorm{F}_p$ denotes $\avg{x\in \{0,1\}^r}{|F(x)|^p}^{1/p}$.

We will use the following hypercontractivity theorem of Polanskiy~\cite{Polyanskiy}.

\begin{theorem}[Follows from Theorem 1 in~\cite{Polyanskiy}]
\label{thm:polanskiy}
Assume that $\eta \in [1/20,19/20]$ and $\eta_0 = 1/20$. For any $F\in \Fn{r}$, we have
\[
\mynorm{T_\eta F}_2 \leq C\cdot \mynorm{F}_p
\]
for $p = 1+(1-2\eta_0)^2$ and $C$ is an absolute constant.
\end{theorem}

\begin{corollary}
\label{cor:hyper}
Assume that $\eta_0,\eta$ are as in the statement of Theorem~\ref{thm:polanskiy} and let $\delta\in (0,1)$ be arbitrary. Say $E\subseteq \{0,1\}^r$ s.t. $|E|\leq \delta\cdot 2^r$. Assume that $(x',x'')\in \{0,1\}^r$ are chosen as follows: $x'\in \{0,1\}^r$ and $I'\in \binom{[r]}{\eta r}$ are chosen i.u.a.r., and we set $x'' = x'\oplus I'$. Then we have
\[
\prob{x',I'}{x'\in E \wedge x''\in E} \leq C\cdot \delta^{1+(1/40)}
\]
where $C$ is the constant from Theorem~\ref{thm:polanskiy}.
\end{corollary}

\begin{proof}
Let $F:\{0,1\}^n\rightarrow \{0,1\}\subseteq \mathbb{R}$ be the indicator function of the set $E$. Note that we have
\[
\prob{x',I'}{x'\in E \wedge x''\in E} = \avg{x',I'}{F(x')F(x'\oplus I')} = \avg{x'}{F(x')T_\eta F(x')}.
\] 

By the Cauchy-Schwarz inequality and Theorem~\ref{thm:polanskiy} we get
\begin{equation}
\label{eq:hyper}
\avg{x'}{F(x')T_\eta F(x')} \leq \mynorm{F}_2\cdot C\cdot \mynorm{F}_p
\end{equation}
for $p = 1+(1-2\eta_0)^2$. Note that we have 
\begin{align*}
\mynorm{F}_p &\leq \delta^{1/p} = \delta^{\frac{1}{1+(1-2\eta_0)^2}}\\
& = \delta^{\frac{1}{2(1-2\eta_0(1-\eta_0))}}\leq (\sqrt{\delta})^{1+\min\{\eta_0,1-\eta_0\}} = \sqrt{\delta}^{1+(1/20)}
\end{align*}
where for the last inequality we have used the fact that for $\eta_0\in [0,1]$ we have
\[
\frac{1}{1-2\eta_0(1-\eta_0)} \geq 1+2\eta_0(1-\eta_0) \geq 1+\min\{\eta_0,1-\eta_0\}.
\]

Putting the upper bound on $\mynorm{F}_p$ together with the fact that $\mynorm{F}_2 \leq \sqrt{\delta}$ and (\ref{eq:hyper}), we get the claim.
\end{proof}

\section{Results}
\label{sec:results}

We show upper and lower bounds for testing and decoding polynomial codes over grids. All our upper bounds hold in the non-adaptive setting, while our lower bounds hold in the stronger adaptive setting. 


Our first result is that for any choice of the field $\F$ (possibly even infinite), the space of functions $\Fnd{n}{d}$ is locally testable. More precisely, we show the following.

\begin{theorem}[$\Fnd{n}{d}$ has a local tester for any field]
\label{thm:test-main-intro}
Let $\F$ be any field. Fix a positive integer $d$ and any $n\in \mathbb{N}$. Then the space $\mc{F}(n,d;\F)$ has a non-adaptive $(\varepsilon,q)$-local tester for $q = 2^{O(d)}\cdot \poly(1/\varepsilon)$.
\end{theorem}

In contrast, we show that the space $\Fnd{n}{d}$ is \emph{not} locally \emph{decodable} over fields of large characteristic, even for $d=1$.

\begin{theorem}[$\Fnd{n}{d}$ does not have a local decoder for large characteristic]
\label{thm:imposs-large-char-intro}
Let $n\in \mathbb{N}$ be a growing parameter. Let $\F$ be any field such that either $\charF(\F) = 0$ or $\charF(\F)\geq n^2$. Then any adaptive $(\varepsilon,q)$-local decoder for $\mc{F}(n,1;\F)$ that corrects an $\varepsilon$ fraction of errors must satisfy $q = \Omega_\varepsilon(\log n/\log \log n)$. 
\end{theorem}

Complementing the above result, we can show that if $\charF(\F)$ is a constant, then in fact the space $\Fnd{n}{d}$ does have a local decoding procedure. 

\begin{theorem}[$\Fnd{n}{d}$ has a local decoder for constant characteristic]
\label{thm:decoding-small-char-intro}
Let $\charF(\F) = p$ be a positive constant. Fix any $d,n\in \mathbb{N}$. There is a $k\leq pd$ such that the space $\mc{F}(n,d;\F)$ has a non-adaptive $(1/2^{O(k)},4^k)$-local decoder.
\end{theorem}

Finally, we augment our testing theorem to get a tolerant tester for $\Fnd{n}{d}$.
\begin{theorem}
\label{thm:tolerant-Fq-intro}
For any $d\in \mathbb{N}$, there is a fixed $\delta = \delta(d)\in (0,1)$ such that for every $\delta_1 < \delta_2 < \delta$, the space $\Fnd{n}{d}$ has a $(\delta_1,\delta_2,q)$-tolerant local tester, where $q = O_{d,\delta_1,\delta_2}(1) = \left(\frac{O(d)}{(\delta_2-\delta_1)^4}\right)^d$.
\end{theorem}

\section{A local tester for $\Fnd{n}{d}$ over any field}
\label{sec:ltall}

We now present our local tester and its analysis. The reader may find the overview from Section~\ref{ssec:overview} helpful while reading the below.

We start by introducing some notation for this section. Throughout, fix any field $\F.$ We consider functions $f:\{0,1\}^I\rightarrow \F$ where $I$ is a finite set of positive integers and indexes into the set of variables $\{X_i\ |\ i\in I\}$. We denote this space as $\Fn{I}.$ Similarly, $\Fnd{I}{d}$ is defined to be the space of functions of degree at most $d$ over the variables indexed by $I$.

The following is the test we use to check if a given function $f:\{0,1\}^I\rightarrow \mathbb{F}$ is close to $\Fnd{I}{d}$.

\paragraph{Test $T_{k,I}(f_I)$}
\paragraph{Notation.} Given two variables $X$ and $Y$ and $a\in \{0,1\}$, ``replacing $X$ by $a\oplus Y$'' refers to substituting $X$ by $Y$ if $a=0$ and by $1-Y$ if $a = 1$.

\begin{itemize}
\item If $|I| > k$, then
\begin{itemize}
\item Choose a random $a\in \{0,1\}$ and distinct $i_0,j_0\in I$ at random and replace $X_{j_0}$ by $a\oplus X_{i_0}$. Let $f'_I$ denote the resulting restriction of $f_I$.
\item Run $T_{k,I\setminus \{j_0\}}(f'_I)$ and output what it outputs.
\end{itemize}
\item If $|I| = k$ then
\begin{itemize}
\item Choose a uniformly random bijection $\sigma: I\rightarrow [k]$.
\item Choose an $a\in \{0,1\}^{k}$ uniformly at random.
\item Replace each $X_i$ ($i\in I$) with $Y_{\sigma(i)}\oplus a_i$.
\item Check if the restricted function $g(Y_1,\ldots,Y_k)\in \Fnd{k}{d}$ by querying $g$ on all its inputs. Accept if so and reject otherwise.
\end{itemize}
\end{itemize}

\begin{remark}
\label{rem:rand-bij}
It is not strictly necessary to choose a \emph{random} bijection $\sigma$ in the test $T_{k,I}$ and a fixed bijection $\sigma: I\rightarrow [k]$ would do just as well. However, the above leads to a cleaner reformulation of the test in Section~\ref{sec:SDL} below.
\end{remark}

\begin{observation}
\label{obs:queries}
Test $T_{k,I}$ has query complexity $2^k$.
\end{observation}

\begin{observation}
\label{obs:completeness}
If $f_I\in \Fnd{I}{d}$, then $T_{k,I}$ accepts with probability $1$.
\end{observation}

The following theorem is the main result of this section and implies Theorem~\ref{thm:test-main-intro} from Section~\ref{sec:results}.

\begin{theorem}
\label{thm:test-main}
For each positive integer $d$, there is a $k  = O(d)$ and $\varepsilon_0 = 1/2^{O(d)}$ such that for any $I$ of size at least $k+1$ and any $f_I\in \mc{F}(I)$,
\[
\prob{}{\text{Test $T_{k,I}$ rejects $f_I$}} \geq \frac{1}{2^{O(d)}}\cdot \min\{\delta_d(f_I),\varepsilon_0\}.
\]
\end{theorem}

Theorem~\ref{thm:test-main-intro} immediately follows from Theorem~\ref{thm:test-main} since to get an $(\varepsilon, 2^{O(d)})$-tester, we repeat the test $T_{k,[n]}$ $t = 2^{O(d)}\cdot \poly(1/\varepsilon)$ many times and accept if and only if each iteration of the test accepts. If the input function $f\in \Fn{n}$ is of degree at most $d$, this test accepts with probability $1$. Otherwise, this test rejects with probability at least $3/4$ for suitably chosen $t$ as above. The number of queries made by the test is $2^{k}\cdot t = 2^{O(d)}\cdot \poly(1/\varepsilon).$

\paragraph{Parameters.} For the rest of this section, we use the following parameters. We choose 
\begin{equation}
\label{eq:choicek}
k = M \cdot d 
\end{equation}
for a large absolute constant $M\in \mathbb{N}$ and set 
\begin{equation}
\label{eq:choiceeps}
\varepsilon_1 = \frac{1}{(4C\cdot 2^{k\cdot H(1/M)})^{40}}
\end{equation}
where $C$ is the absolute constant from Corollary~\ref{cor:hyper}. The constant $M$ is chosen so that 
\begin{equation}
\label{eq:choiceM}
H(1/M) < \frac{1}{20} \qquad\qquad\qquad\text{and}\qquad\qquad\qquad k \geq 100 \log \frac{2}{\varepsilon_1}.
\end{equation}
Note that the second constraint is satisfied for a large enough absolute constant $M$ since we have
\[
\frac{100\log(2/\varepsilon_1)}{k} \leq \frac{40kH(1/M) + 40\log C + O(1)}{k} \leq 40H(1/M) + \frac{40\log C + O(1)}{M}
\]
which can be made arbitrary small for large enough constant $M$.
 Further, we set 
 \begin{equation}
 \label{eq:choiceelleps0}
  \ell = \log \frac{2}{\varepsilon_1}\qquad\qquad \text{and}\qquad\qquad\varepsilon_0 = \frac{\varepsilon_1}{100\ell}.
 \end{equation}

The following are the two main lemmas used to establish Theorem~\ref{thm:test-main}.

\begin{lemma}[Small distance lemma]
\label{lem:SDL}
Fix any $I$ such that $|I| = r > k+1$ and $f_I:\{0,1\}^I\rightarrow\mathbb{F}$ such that $\delta_d(f_I) = \delta \leq \varepsilon_1$.
\[
\prob{}{\text{$T_{k,I}$ rejects $f_I$}} \geq \frac{\delta}{2^{O(k)}}.
\]
\end{lemma}

\begin{lemma}[Large distance lemma]
\label{lem:LDL}
Fix any $I$ such that $|I| = r$ satisfies $r^2 > 100\ell^2$ and $f_I:\{0,1\}^I\rightarrow\mathbb{F}$ such that $\delta_d(f_I) > \varepsilon_1$. Then
\[
\prob{}{\delta_d(f'_I) < \varepsilon_0} < \frac{100\ell^2}{r^2}.
\]
\end{lemma}

With the above lemmas in place, we show how to finish the proof of Theorem~\ref{thm:test-main}.

\begin{proof}[Proof of Theorem~\ref{thm:test-main}]
Fix any $I$ and consider the behaviour of the test $T_{k,I}$ on $f_I$. Assume $|I| = n$.

A single run of $T_{k,I}$ produces a sequence of functions $f_n = f_I,f_{n-1},\ldots,f_{k}$, where $f_r$ is a function on $r$ variables. Let $I_n = I,I_{n-1},\ldots,I_{k}$ be the sequence of index sets produced. We have $f_r:\{0,1\}^{I_r}\rightarrow \mathbb{F}$. Note that $k\geq 100\ell$ by (\ref{eq:choiceM}) and (\ref{eq:choiceelleps0}).

Define the following pairwise disjoint events for each $r\in \{k,\ldots, n\}$.
\begin{itemize}
\item $\mf{F}_r$ is the event that $\delta_d(f_r) > \varepsilon_1$.
\item $\mf{C}_r$ is the event that $\delta_d(f_r) < \varepsilon_0$.
\item $\mc{E}_r$ is the event that $\delta_d(f_r) \in [\varepsilon_0,\varepsilon_1]$.
\end{itemize}

For any $f_I$, one of $\mf{F}_r,\mf{C}_r,$ or $\mc{E}_r$ occurs with probability $1$. If either $\mc{E}_n$ or $\mf{C}_n$ occurs, then by Lemma~\ref{lem:SDL} we are done. Therefore, we assume that $\mf{F}_n$ holds.

We note that one of the following possibilities must occur: either all the $f_r$ satisfy $\delta_d(f_r) > \varepsilon_1$; or there is some $f_r$ such that $\delta_d(f_r) \in [\varepsilon_0,\varepsilon_1]$; or finally, there is some $f_r$ such that $\delta_d(f_{r+1}) > \varepsilon_1$ but $\delta_d(f_r) < \varepsilon_0$. We handle each of these cases somewhat differently.

Clearly, if $\mf{F}_{k}$ holds, then $\deg(f_{I_k}) > d$ and hence $T_{k,I_{k}}$ rejects $f_{I_{k}}$ with probability $1$. On the other hand, by Lemma~\ref{lem:SDL}, we see that
\[
\prob{}{\text{$T_{k,I}$ rejects $f_I$}\ |\ \bigvee_{r = k}^{n-1} \mc{E}_r} \geq \frac{\varepsilon_0}{2^{O(k)}}.
\]

Thus, we have
\begin{equation}
\label{eq:rej-prob}
\prob{}{\text{$T_{k,I}$ rejects $f_I$}} \geq \frac{\varepsilon_0}{2^{O(k)}}\cdot \prob{}{\bigvee_{r = k}^{n-1} \mc{E}_r \vee \bigwedge_{r = k}^{n-1} \mf{F}_r}
\end{equation}

Let $\mc{E}$ denote the event $\neg(\bigvee_{r = k}^{n-1} \mc{E}_r \vee \bigwedge_{r = k}^{n-1} \mf{F}_r)$. Notice that if event $\mc{E}$ occurs, there must be an $r \geq k$ such that $\mf{C}_r$ occurs but we also have $\mf{F}_{r+1}\wedge\mf{F}_{r+2}\wedge\cdots\wedge\mf{F}_{n} $. By Lemma~\ref{lem:LDL}, the probability of this is upper bounded by $100\ell^2/r^2$ for each $r\geq k.$

By a conditional probability argument, we see that
\[
\prob{}{\neg\mc{E}} \geq \prod_{r \geq k} \left( 1- \frac{100\ell^2}{r^2}\right) \geq \exp\left(-200\ell^2\sum_{r\geq k}\frac{1}{r^2}\right) \geq \exp(-O(\ell)) = \frac{1}{2^{O(k)}}
\]
where we have used the fact that $k\geq 100\ell$ and for the second inequality we also use  $(1-x) \geq \exp(-2x)$ for $x\in [0,1/2]$. Plugging the above into (\ref{eq:rej-prob}), we get the theorem.
\end{proof}

It remains to prove Lemmas~\ref{lem:SDL} and \ref{lem:LDL} which we do in Sections~\ref{sec:SDL} and \ref{sec:LDL} respectively.

\subsection{Proof of Small Distance Lemma (Lemma~\ref{lem:SDL})}
\label{sec:SDL}

 We start with a brief overview of the proof of Lemma~\ref{lem:SDL}. Suppose $f_I$ is $\delta$-close to some polynomial $P$ for some $\delta \leq \varepsilon_1$. As mentioned in Section~\ref{ssec:overview}, our aim is to show that the (random) restriction $g$ of $f$ obtained above and the corresponding restriction $Q$ of $P$ differ at only one point. Then we will be done since any two distinct degree-$d$ polynomials on $\{0,1\}^k$ must differ on at least $2$ points (if $k > d$) and hence the restricted function $g$ cannot be a degree-$d$ polynomial. 

Note that the restriction is effectively given by $a \in \{0,1\}^I$ and $\phi:I\to [k]$
such that $g(y) = f_I(x(y))$ where $x(y) = (x_i(y))_{i \in I}$ is given by $x_i(y_1,\ldots,y_k) = y_{\phi(i)} \oplus a_i$. ($\phi$ is obtained by a sequence of replacements followed by the bijection $\sigma$.) Similarly we define $Q(y) = P(x(y))$.  To analyze the test, 
 we consider the queries $\{x(y)\}_{y \in \{0,1\}^k}$ made to the oracle for $f_I$. For every fixed $y\in \{0,1\}^k$ the randomness (in $a$ and $\phi$) leads to a random query $x(y)\in \{0,1\}^I$ to $f_I$ and it is not hard to show that for each fixed $y$, $x(y)$ is uniformly distributed over  $\{0,1\}^I.$ Hence, the probability that $g$ and $Q$ differ at any fixed $y\in \{0,1\}^k$ is exactly $\delta$. 

We would now like to say that for distinct $y',y''\in \{0,1\}^k$, the probability that $g$ and $Q$ differ at \emph{both} $y'$ and $y''$ is much smaller than $\delta$. This would  be true if, for example, $x(y')$ and $x(y'')$ were independent of each other, but this is unfortunately not the case. For example, consider the case when no $X_i$ ($i\in I$) is identified with the variable $Y_k$ (i.e., for every $i \in I$, $\phi(i) \ne k$).\footnote{Strictly speaking this case can not occur due to the way $\phi$ is constructed, but it is is useful to think about this case anyway.} In this case, $x(y') = x(y'')$ for every $y'$ and $y''$ that differ only at the $k$th position. More generally, if the number of variables that are identified with $Y_k$ is very small (much smaller than the expected number $r/k$) then $x(y')$ and $x(y'')$ would be heavily correlated if $y'$ and $y''$ differed in only the $k$th coordinate.

So, the first step in our proof is to analyze the above restriction process and show that with reasonable probability, for every $Y_j$ there are many variables (close to the expected number) mapped to it, i.e., $|\phi^{-1}(j)|$ is $\Omega(r/k)$ for every $j \in [k]$. To get to this analysis we first give an alternate (non-iterative) description of the test $T_{k,I}$ and analyze it by exploring the random set-union process mentioned in Section~\ref{ssec:overview}. We note that this process and its analysis may be independently interesting. 

Once we have a decent lower bound on $\min_j |\phi^{-1}(j)|$, we can use the hypercontractivity theorem of Polyanskiy (Theorem~\ref{thm:polanskiy}) to argue that for any $y'\neq y''$, the inputs $x(y')$ and $x(y'')$ are somewhat negatively correlated (see Corollary~\ref{cor:hyper}).
We note that since the distribution of the pair $(x(y'),x(y''))$ is not the usual noisy hypercube distribution and so the usual hypercontractivity does not help. But this is where the strength of Polyanskiy's hypercontractivity comes in handy --- even after we fix the Hamming distance between $x(y')$ and $x(y'')$ the symmetry of the space leads to enough randomness to apply Theorem~\ref{thm:polanskiy}. This application already allows us to show a
weak version of Lemma~\ref{lem:SDL} and hence a weak version of our final tester. 

To prove Lemma~\ref{lem:SDL} in full strength as stated, we note that stronger parameters for the lemma are linked to stronger negative correlation between $x(y')$ and $x(y'')$ for various $y'$ and $y''$. It turns out that this is directly related to the Hamming distance of $y'$ and $y''$: specifically, we would like their Hamming distance to not be too close to $0$ or to $k$. Hence, we would like to restrict our attention to a subset $T$ of the query points of $\{0,1\}^k$ that form such a ``code''. At the same time, however, we need to ensure that, as for $\{0,1\}^k$, any two distinct degree-$d$ polynomials cannot differ at exactly one point in $T$. We construct such a set $T$  in Claim~\ref{clm:Texists}, and use it to prove Lemma~\ref{lem:SDL}. 

\vspace*{5pt}

We now begin the formal proof with an alternate
 but equivalent (non-recursive) description of test $T_{k,I}$ for $|I| = r > k$.

\paragraph{Test $T_{k,I}$} (Alternate description)
\begin{itemize}
\item Choose $a\in \{0,1\}^r$ uniformly at random.
\item Choose a bijection $\pi:[r]\rightarrow I$ uniformly at random.
\item Choose $p:\{k+1,\ldots,r\}\rightarrow \mathbb{Z}$ so that each $p(i)$ is uniformly distributed over the set $\{1,\ldots,i-1\}$ and the $p(i)$s are mutually independent. (Here $p(i)$ stands for the ``parent of $i$'').
\item For $i$ in $r,r-1,\ldots,k+1$
\begin{itemize}
\item Substitute $X_{\pi(i)}$ by $a_i \oplus X_{\pi(p(i))}$.
\end{itemize}
\item For $i\in 1,\ldots,k$
\begin{itemize}
\item Replace each $X_{\pi(i)}$ with $a_i\oplus Y_{i}$ for each $i\in [k]$.
\end{itemize}
\item Check if the restricted function $g(Y_1,\ldots,Y_k)$ is of degree at most $d$ by querying $g$ on all its inputs.
Accept if so and reject otherwise.
\end{itemize}

\begin{proposition}
\label{prop:equiv}
The iterative description above is equivalent to test $T_{k,I}$.
\end{proposition}

We now begin the analysis of the test $T_{k,I}.$ As stated above, the first step is to understand the distribution of the number of $X_i$ ($i\in I$) eventually identified with $Y_j$ (for various $j\in [k]$). We will show (Corollary~\ref{cor:pdist}) that with reasonable probability, each $Y_j$ has $\Omega(r/k)$ $X_i$s that are identified with it. 

Fix any bijection $\pi:[r]\rightarrow [r]$. For $i,j$ such that $i\geq j$ and $i\in \{k,\ldots,r\}$, we define $B_{j, i}$ to be the index set of those variables that are identified with $X_{\pi(j)}$ (or its complement) in the first $r-i$ rounds of substitution. Formally,
\[
B_{j,i} = \left\{
\begin{array}{ll}
\{\pi(j)\} & \text{if $i=r$,}\\
B_{j,i+1}& \text{if $i<r$ and $p(i+1) \neq j$.}\\
B_{j,i+1}\cup B_{i+1,i+1} & \text{if $i<r$ and $p(i+1) = j$.}
\end{array}\right.
\]
For $j\in [k]$, let $B_j = B_{j,k}$. This is the set of $i$ such that $X_{\pi(i)}$ is ``eventually'' identified with $X_{\pi(j)}$ (or its complement). For $i\in [r]$, we define $b(i) = j$ if $i\in B_j$.

To analyze the distribution of the ``buckets'' $B_1,\ldots, B_k$, it will be helpful to look at an equivalent way of generating this distribution. We do this by sampling the buckets in ``reverse'': i.e., we start with the $j$th bucket being the singleton set $\{j\}$ and for each $i = k+1,\ldots,r$, we add $i$ to the $j$th bucket if $i$ falls into the the $j$th bucket.

Formally, for each $j\in [k]$, define the set $B'_{j,i}$ to be $B_{j}\cap [i]$. Note that we have
\[
B'_{j,i+1} = \left\{
\begin{array}{ll}
\{j\} & \text{if $i=k$,}\\
B'_{j,i} & \text{if $i > k$ and $p(i+1) \not\in B'_{j,i}$.}\\
B'_{j,i}\cup \{i+1\} & \text{if $i > k$ and $p(i+1) \in B'_{j,i}$.}
\end{array}\right.
\]
In particular, we see that for any $i \geq k+1$,
\begin{equation}
\label{eq:Bji}
\prob{p}{\pi(i+1)\in B'_{j,i+1}\ |\ B'_{1,i},\ldots,B'_{k,i}} = \prob{p}{p(i+1)\in B'_{j,i}\ |\ p(k+1),\ldots,p(i)} = \frac{|B'_{j,i}|}{i}.
\end{equation}

This yields the following equivalent way of sampling sets from the above distribution.
\begin{lemma}
\label{lem:pdist}
Consider the following sampling algorithm that partitions $[r]$ into $k$ parts. Choose a random  permutation $\sigma$ of the set $\{1,\ldots,r\}$ as follows. First choose a uniform element $i_1\in \{1,\ldots,k\}$. Now choose a uniformly random permutation $\sigma$ of $[r]$ such that $\sigma(1) = 1$ and  assume that the elements of $[k]\setminus \{i_1\}$ appear in the order $i_2,\ldots,i_{k}$ in $\sigma$ ($\sigma(i_2)<\cdots < \sigma(i_k)$). Define $C_1,\ldots,C_k$ as follows: 
\begin{itemize}
\item $C_{1} = \{j > k\ |\  \sigma(j) < \sigma(i_2)\}\cup \{1\}$,
\item $C_{2} = \{j > k\ |\ \sigma(i_2) < \sigma(j) < \sigma(i_3)\}\cup \{2\},$
\item $\ldots$
\item $C_{k-1} = \{j > k\ |\ \sigma(i_{k-1}) < \sigma(j) < \sigma(i_k)\}\cup \{k-1\},$
\item $C_k = \{j > k\ |\ \sigma(i_{k}) < \sigma(j)\}\cup \{k\}$.
\end{itemize}
Then the distribution of $(C_1,\ldots,C_{k})$ is identical to the distribution of $(B_1,\ldots,B_{k})$.
\end{lemma}

\begin{proof}
Assume $\sigma$ is sampled by starting with the element $1$ and then inserting the elements $i=2,\ldots,r$ one by one in a random position \emph{after} $1$ (since we are sampling $\sigma$ such that $\sigma(1) = 1$). Simultaneously, consider the evolution of the $j$th bucket. Let $C_{j,i}$ denote the $j$th bucket after elements $2,\ldots,i$ have been inserted. 

Note that no matter how the first $k$ elements are ordered in $\sigma$, the element $j\in [k]$ goes to the $j$th bucket at the end of the sampling process. Thus, after having inserted $2,\ldots,k$, we have $C_{j,k} = \{j\}.$

We now insert $(i+1)$ for each $i$ such that $k\leq i < r$. The position of $i+1$ is a uniform random position after the first position. For each $i$, the probability that $i+1$ ends up in the $j$th bucket can be seen to be $|C_{j,i}|/i$, exactly as in (\ref{eq:Bji}). This shows that $(C_1,\ldots, C_k)$ has the same distribution as $(B_1,\ldots,B_k).$
\end{proof}

\begin{corollary}
\label{cor:pdist}
With probability at least $\frac{1}{2^{O(k)}}$ we have $|B_j|\geq \frac{r}{4k}$ for each $j\in [k]$.
\end{corollary}

\begin{proof}
We assume that $r > 4k$ since otherwise the statement to be proved is trivial (as each $|B_j|\geq 1$ with probability $1$.) By Lemma~\ref{lem:pdist} it suffices to prove the above statement for the sets $(C_1,\ldots,C_{k})$. 

Now, say a permutation $\sigma$ of $[r]$ fixing $1$ is chosen u.a.r. and we set $C_j$ as in Lemma~\ref{lem:pdist}. We view the process of sampling $\sigma$ as happening in two stages: we first choose a random linear ordering of $A=\{k+1,\ldots,r\}$, i.e. a random function $\sigma':A\rightarrow [r-k]$, and then inserting the elements $2,\ldots,k$ one by one at random locations in this ordering. (The position of the element $1$ is of course determined.)

Condition on any choice of $\sigma'$. For $j\in \{2,\ldots,k\}$, let $C_j' = \{i\ |\ (j-1)r/k \leq \sigma'(i) \leq (j-1)r/k + \lceil r/2k\rceil\}$. Fix any bijection $\tau:\{2,\ldots,k\}\rightarrow \{2,\ldots,k\}.$ 

Consider the probability that on inserting $2,\ldots,k$ into the ordering $\sigma'$, each $j\in \{2,\ldots,k\}$ is inserted between two elements of $C_{\tau(j)}'$. Call this event $\mc{E}_\tau.$ Conditioned on this event, it can be seen that for each $j\in \{2,\ldots,k\}$, the $j$th bucket $C_j$ has size at least 
\[
\min\{a\ |\ a\in  C'_j\} - \max \{a\ |\ a\in  C'_{j-1}\} \geq \frac{jr}{k}-(\frac{(j-1)r}{k} + \frac{r}{2k}+1) = \frac{r}{2k}-1\geq \frac{r}{4k}
\]
where we have defined $C'_1 = \{0\}$. Similarly, conditioned on $\mc{E}_\tau$, we have $|C_1|\geq r/k\geq r/(4k).$

Since this holds for each $\tau$ and the events $\mc{E}_\tau$ are mutually exclusive, we have
\[
\prob{}{\forall j\in [k],\ |C_j| \geq \frac{r}{4k}} \geq \sum_{\tau} \prob{}{\mc{E}_\tau}.
\]

We now analyze $\prob{}{\mc{E}_\tau}$ for any fixed $\tau$. Conditioned on the positions of $2,\ldots,j-1$, the probability that $\sigma(j)\in C'_{\tau(j)}$ is at least $ (r/(2k))\cdot (1/r) = 1/(2k)$. Therefore we have
\[
\prob{}{\mc{E}_\tau} \geq 1/2^{k-1}k^{k-1}.
\]
 Thus, we get

\[
\prob{}{\forall j\in \{2,\ldots,k\},\ |C_j| \geq \frac{r}{4k}} \geq \sum_{\tau} \prob{}{\mc{E}_\tau} \geq \frac{(k-1)!}{2^{k-1}\cdot k^{k-1}} \geq \frac{1}{2^{O(k)}},
\]
where we have used the Stirling approximation for the final inequality. This concludes the proof of the corollary.
\end{proof}

Note that the sets $B_j$ are determined by our choice of $p$. For the rest of the section, we condition on a choice of $p = p_0$ such that Corollary~\ref{cor:pdist} holds. We now show how to finish the proof of Lemma~\ref{lem:SDL}. 

Fix a polynomial $P\in \Fnd{I}{d}$ such that $\delta(f_I,P) = \delta_d(f_I) = \delta$ as in the lemma statement. Let $E\subseteq \{0,1\}^I$ be the set of points where $f$ and $P$ differ. We have
$
\frac{|E|}{2^r} = \delta \leq \varepsilon_1.
$

For $y',y''\in \{0,1\}^k$, we use $\Delta(y',y'')$ to denote the Hamming distance between them and $\Delta'(y',y'')$ to denote the quantity $\min \{\Delta(y',y''),k-\Delta(y',y'')\}$.

We prove the following two claims. 

\begin{claim}
\label{clm:Texists}
There is a non-empty set $T\subseteq \{0,1\}^k$ such that:
\begin{itemize}
\item $|T|\leq \binom{k}{\leq d} + 1$, 
\item Given distinct $y',y''\in T$, $\Delta'(y',y'')\geq k/4$,
\item No pair of polynomials $P,P'\in \Fnd{I}{d}$ can differ at exactly one input from $T$.\footnote{Note that it could be that two distinct polynomials in $\Fnd{I}{d}$ agree everywhere in $T$.}
\end{itemize}
\end{claim}

For each input $y\in \{0,1\}^k$ to the restricted polynomial $g$, let $x(y)\in \{0,1\}^{I}$ be the corresponding input to $f_I$. Let $S$ denote the multiset $\{x(y)\ |\ y\in T\}.$ This is a subset of the set of inputs on which $f_I$ is queried. 

\begin{claim}
\label{clm:onept}
Let $p=p_0$ be as chosen above. With probability at least $\delta\cdot (|T|/2)$ over the choice of $\pi$ and $a$, we have $|S\cap E| = 1$ (i.e. there is a unique $y\in T$ such that $x(y)\in E$).
\end{claim}

Assuming Claims~\ref{clm:Texists} and \ref{clm:onept}, we have proved Lemma~\ref{lem:SDL} since with probability at least $\frac{1}{2^{O(k)}}\cdot \delta\cdot (|T|/2)$ (cf. Corollary~\ref{cor:pdist} and Claim~\ref{clm:onept}), the restricted function $g(Y_1,\ldots,Y_k)$ differs from the restriction $P'(Y_1,\ldots,Y_k)$ of $P$ at exactly $1$ point in $T$. However, by our choice of the set $T$, any two polynomials from $\Fnd{k}{d}$ that differ on $T$ must differ on at least two inputs. Hence, $g$ cannot be a degree $d$ polynomial, and thus the test rejects.

\subsubsection{Proof of Claim~\ref{clm:Texists}}

Given functions $f,g\in \Fn{k}$, we define their inner product $\ip{f}{g}$ by $\ip{f}{g} = \sum_{y\in \{0,1\}^k}f(y)g(y)$. Recall that $\Fnd{k}{d}^\perp$ is defined to be the set of all $f\in \Fn{k}$ such that $\ip{f}{g} = 0$ for each $g\in \Fnd{k}{d}.$

We will construct $T$ by finding a suitable non-zero $f\in\Fnd{k}{d}^\perp$ and setting $T = \Supp(f)$, where $\Supp(f) = \{y\in \{0,1\}^k\ |\ f(y) \neq 0\}$. Thus, we need $f$ to satisfy the following properties.
\begin{enumerate}
\item $|\Supp(f)|\leq \binom{k}{\leq d} + 1$, 
\item Given distinct $y',y''\in \Supp(f)$, $\Delta'(y',y'')\geq k/4$,
\item No pair of polynomials $P,P'\in \Fnd{I}{d}$ can differ at exactly one input from $\Supp(f)$.
\end{enumerate}

We first observe that Property 3 is easily satisfied. To see this, assume that $g_1,g_2\in \Fnd{k}{d}$ differ at exactly one point, say $y'$, from $\Supp(f)$. Then, since $g = g_1-g_2\in \Fnd{k}{d}$ and $f\in \Fnd{k}{d}^\perp$, we must have $\ip{f}{g} = 0$. On the other hand since $\Supp(g)\cap \Supp(f) = \{y'\}$, we have 
\[
\ip{f}{g} = \sum_{y\in \{0,1\}^k} f(y)g(y) = f(y')g(y') \neq 0
\]
which yields a contradiction. Hence, we see that $g_1$ and $g_2$ cannot differ at exactly one point in $\Supp(f).$

We thus need to choose a non-zero $f\in \Fnd{k}{d}^\perp$ so that Properties 1 and 2 hold. Note that to ensure that $f\in \Fnd{k}{d}^\perp$, it suffices to ensure that for each $A\subseteq [k]$ of size at most $d$ we have
\begin{equation}
\label{eq:ipfA}
\sum_{y\in \{0,1\}^k}f(y)\cdot \prod_{i\in A}y_i = 0.
\end{equation}
The number of such $A$ is $N = \binom{k}{\leq d}.$

To ensure that Properties 1 and 2 hold, it suffices to ensure that $\Supp(f)\subseteq U$ where $U\subseteq \{0,1\}^{k}$ is a set of size $N+1$ so that any distinct $y',y''\in U$ satisfy $\Delta(y',y'')\in [k/4,3k/4]$. (Note that this implies that $\Delta'(y',y'')\geq k/4.$)

To see that such a set $U$ exists, consider the following standard greedy procedure for finding such a set $U$: starting with an empty set, we repeatedly choose an arbitrary point $z$ to add to $U$ and then remove all points at Hamming distance at most $k/4$ and at least $3k/4$ from $z$ from future consideration. Note that this procedure can produce up to $2^{k}/(2\binom{k}{\leq k/4})$ many points. By Fact~\ref{fac:Hamballs} and our choice of $k$ (see (\ref{eq:choicek}) and (\ref{eq:choiceM})) we have
\begin{align*}
\frac{2^{k}}{2\binom{k}{\leq k/4}} &\geq 2^{k(1-H(1/4))-1} \geq 2^{k/20}\\
N = \binom{k}{\leq d} &\leq 2^{kH(d/k)}  <   2^{k/20}.
\end{align*}
Hence, the above greedy procedure can be used to produce a set $U$ of size $N+1$ as required. 

Since we assume that $\Supp(f)\subseteq U$, ensuring (\ref{eq:ipfA}) reduces to ensuring the following for each $A\subseteq [k]$ of size at most $d$:
\begin{equation}
\label{eq:ipfA2}
\sum_{y\in U}f(y)\cdot \prod_{i\in A}y_i = 0.
\end{equation}
Choosing $f(y)$ ($y\in U$) so that the above holds reduces to solving a system of $N$ homogeneous linear equations (one for each $A$) with $|U| = N+1$ constraints. By standard linear algebra, this system has a non-zero solution. This yields a non-zero $f\in \Fnd{k}{d}^\perp$ with the required properties.

\subsubsection{Proof of Claim~\ref{clm:onept}}

Let $y',y''$ be any two distinct points in $T$. Let $\Delta$ denote $\Delta(y',y'')$ and $\Delta'$ denote $\Delta'(y',y'')$. We show that
\begin{align}
\prob{\pi,a}{x(y')\in E} &= \delta\label{eq:l1incenc}\\
\prob{\pi,a}{x(y')\in E \wedge x(y'')\in E}&\leq C\cdot \delta^{1+(1/40)}.\label{eq:l2incenc}
\end{align}
where $C$ is the absolute constant from the statement of Corollary~\ref{cor:hyper}.

Given (\ref{eq:l1incenc}) and(\ref{eq:l2incenc}) we are done since we can argue by inclusion exclusion as follows.
\begin{align*}
\prob{\pi,a}{|S\cap E| = 1} &\geq \sum_{y\in T} \prob{}{x(y)\in E} - \sum_{y'\neq y''\in T} \prob{}{x(y')\in E\wedge x(y'')\in E}\\
&\geq \delta\cdot |T| - |T|^2 \cdot C\cdot  \delta^{1+(1/40)} \qquad\qquad\qquad  (\text{by }(\ref{eq:l1incenc}) \text{ and } (\ref{eq:l2incenc}))\\
&\geq \delta\cdot |T| (1-(\binom{k}{\leq d} + 1)\cdot C\cdot  \varepsilon_1^{1/40})\qquad\ (\because\ \delta \leq \varepsilon_1, |T|\leq \binom{k}{\leq d} + 1)
\end{align*}
Note that by our choice of $\varepsilon_1$ (see~(\ref{eq:choiceeps})) and Fact~\ref{fac:Hamballs} we have
\[
\varepsilon_1^{1/40} \leq \frac{1}{4C2^{kH(d/k)}} \leq \frac{1}{2C\cdot (\binom{k}{\leq d} + 1)},
\]
which along with our previous computation yields
\[
\prob{\pi,a}{|S\cap E| = 1} \geq \delta\cdot \frac{|T|}{2}.
\]

This finishes  the proof of the Claim using (\ref{eq:l1incenc}) and (\ref{eq:l2incenc}). We now prove (\ref{eq:l1incenc}) and (\ref{eq:l2incenc}).

To prove (\ref{eq:l1incenc}), we consider the distribution of $x(y')$ for any fixed $y'\in \{0,1\}^k$.  Condition on any choice of $\pi$. For any $i\in [r]$, let $A_i = \{j\ | i\in \bigcup_{i'} B_{j,i'}\}$. Note that $\pi(j) < \pi(i)$ for each $j\in A_i$. We have
\begin{equation}
\label{eq:xy'}
x(y')_{\pi(i)} = a_{i} \oplus \bigoplus_{j\in A_i} a_j \oplus y'_{b(i)}.
\end{equation}
which is a uniform random bit even after conditioning on all $a_j$ for $j < i$. In particular, it follows that for each choice of $\pi$, $x(y')$ is a uniformly random element of $\{0,1\}^I$. This immediately implies (\ref{eq:l1incenc}). Also note that since $x(y')$ has the same distribution \emph{for each choice of $\pi$}, the random variables $x(y')$ and $\pi$ are independent from each other. 

To prove (\ref{eq:l2incenc}), we will use our corollary to Polyanskiy's Hypercontractivity theorem (Corollary~\ref{cor:hyper}). Let $D\subseteq [k]$ be the set of coordinates where $y'$ and $y''$ differ. Condition on any choice of $x(y')\in \{0,1\}^n$. By (\ref{eq:xy'}), the point $x(y'')$ satisfies, for each $i$,

\[
x(y'')_{\pi(i)} \oplus x(y')_{\pi(i)} = y''_{b(i)}\oplus y'_{b(i)}.
\]
Or equivalently, for any $h\in I$, we have
\[
x(y'')_{h} \oplus x(y')_{h} = y''_{b(\pi^{-1}(h))}\oplus y'_{b(\pi^{-1}(h))}=
\left\{
\begin{array}{ll}
1 & \text{if $\pi^{-1}(h)\in \bigcup_{j\in D}B_j$}\\
0 & \text{otherwise.}
\end{array}
\right.
\]

Now, we may rewrite the condition $\pi^{-1}(h)\in \bigcup_{j\in D}B_j$ as $h\in \pi(B_D)$ for $B_D:=\bigcup_{j\in D}B_j$. Note that $\pi(B_D)\subseteq I$ is a uniformly random subset of $I$ of size $|B_D|$.

Hence, we may equivalently sample the pair $(x(y'),x(y''))$ as follows: Choose $x(y')\in \{0,1\}^{I}$ uniformly at random, and choose independently a random set $I'\subseteq I$ of size $|B_D|$ and flip $x(y')$ exactly in the coordinates in $I'$ to get $x(y'')$. 

Note that $|B_D| = \sum_{j\in D}|B_j|\geq {(\Delta\cdot r)}/{4k}$ since $|D| = \Delta$ and $|B_j|\geq r/4k$ for each $j$ by Corollary~\ref{cor:pdist}. At the same time, we also have $|[r]\setminus B_D| =\sum_{j\in [k]\setminus D} |B_j| \geq (k-\Delta)\cdot r/(4k)$. Thus, $|B_D| = \eta r$ for some $\eta\in [\Delta/(4k),1-(k-\Delta)/(4k)]\subseteq [\Delta'/(4k), 1-(\Delta'/(4k))]$.

By Claim~\ref{clm:Texists}, we know that $\Delta'(y',y'') \geq k/4$ and hence we have $\eta\in [1/16,15/16]$. Applying Corollary~\ref{cor:hyper}, we see that this implies
\[
\prob{x(y'),I'}{x(y')\in E\wedge x(y'')\in E} \leq C\cdot \delta^{1+(1/40)}.
\]
This proves (\ref{eq:l2incenc}) and hence finishes the proof of the claim.

\subsection{Proof of Large Distance Lemma (Lemma~\ref{lem:LDL})}
\label{sec:LDL}

\newcommand{\p}[1]{P^{(#1)}}
\newcommand{\f}[1]{f^{(#1)}}

We follow the proof of~\cite[Lemma 12]{BKSSZ}.

Given a triple $(i,j,b)\in I^2\times \{0,1\}$ with $i,j$ distinct, call $(i,j,b)$ a \emph{bad triple} if the restricted function $f_I'$ obtained when the test chooses $i_0 = i$, $j_0 = j$ and $a=b$ is $\varepsilon_0$-close to $\Fnd{I\setminus j}{d}$. To prove Lemma~\ref{lem:LDL}, it suffices to show that the number of bad triples is at most $100\ell^2$. To do this, we bound instead the number of \emph{bad pairs}, which are defined to be pairs $(i,j)$ for which there exists $b\in \{0,1\}$ such that $(i,j,b)$ is a bad triple. Note that $(i,j)$ is a bad pair iff $(j,i)$ is. Hence, the set of bad pairs $(i,j)$ defines an undirected graph $G_{\text{bad}}$. If there are fewer than $25\ell^2$ edges in $G_{\text{bad}}$, we are done since this implies that there are at most $50\ell^2$ bad pairs and hence at most $100\ell^2$ bad triples. Otherwise, $G_{\text{bad}}$ has more than $25\ell^2$ edges and it is easy to see that one of the following two cases must occur:

\begin{itemize}
\item $G_{\text{bad}}$ has a matching with at least $\ell+1$ edges, or
\item $G_{\text{bad}}$ has a star with at least $\ell+1$ edges.
\end{itemize}

We show that in each case, we can find a polynomial $P\in \Fnd{I}{d}$ is $\varepsilon_1$-close to $f_I$, which will contradict the assumption that $\delta_d(f_I) > \varepsilon_1$ and hence finish the proof of the lemma.

We first note that in either the matching or the star case, we can replace some variables $X$ with $1\oplus X$ in $f_I$ (note that this does not change $\delta_d(f_I)$) to ensure that the bad triples that give rise to the bad pairs are all of the form $(X,X',0)$: i.e., all the bad triples come from identifying variables (and \emph{not} from identifying a variable with the complement of another).

Let $t_1 = (X_{i_1},X_{j_1},0),\ldots,t_{\ell+1}= (X_{i_{\ell+1}},X_{j_{\ell+1}},0)$ denote the bad triples obtained above (in either the matching or the star case). Each triple $t_h$ defines the subset $R_h\subseteq \{0,1\}^{I}$ where the variables $X_{i_h}$ and $X_{j_h}$ take the same values; let $R_h' $ denote the complement of $R_h$. Note that each $|R_h| = 2^{r-1}$. Furthermore, it follows from the form of the triples that for each $h$ we have $|S_1\cap S_2\cap\cdots S_h| = 2^{r-h}$ for any choice of $S_1\in \{R_1,R_1'\},\ldots,S_h\in \{R_h,R_h'\}$.

By assumption, for each triple $t_h$, there is a polynomial $\p{h}$ such that $\p{h}$ is $\varepsilon_0$-close to $\f{h}$, where the latter function is obtained by identifying the variables $X_{i_h}$ and $X_{j_h}$ in $f_I$. We will show the following claim.

\begin{claim}
\label{clm:global}
There is a $P\in \Fnd{I}{d}$ such that for each $h\in [\ell+1]$, $P(x) = \p{h}(x)$ for all $x\in R_h$.
\end{claim}

Assuming the above claim, we show that the polynomial $P$ above is actually $\varepsilon_1$-close to $f_I$, which contradicts our assumption about $\delta_d(f_I)$.

Consider a uniformly random input $x\in \{0,1\}^n$. We have
\begin{align}
\prob{x}{f_I(x) \neq P(x)} &\leq \sum_{h=1}^{\ell} \prob{x}{f_I(x)\neq P(x)\ |\ x\in R_h\setminus \bigcup_{h' < h}R_{h'}}\cdot \prob{}{x\in R_h\setminus\bigcup_{h'< h}R_{h'}}\notag\\ &+ \prob{x}{x\not\in \bigcup_{h\leq \ell}R_h}\label{eq:fPdist}
\end{align}

For each $h$, we have
\begin{align*}
\prob{x}{x\in R_h\setminus \bigcup_{h'< h}R_{h'}} &= \prob{x}{x\in R_h\cap R_1'\cap\cdots \cap R_{h-1}'}= \frac{1}{2^{h}}\\
\prob{x}{x\in R_h\setminus \bigcup_{h'< h}R_{h'}\ |\ x\in R_h} &= \frac{\prob{x}{x\in R_h\cap R_1'\cap\cdots \cap R_{h-1}'}}{\prob{x}{x\in R_h}} = \frac{1}{2^{h-1}}\\
\prob{x}{x\not\in \bigcup_{h\leq \ell}R_h} &= \prob{x}{x\in R_1'\cap\cdots \cap R_{\ell}'} = \frac{1}{2^{\ell}}
\end{align*}

Since $P(x)$ agrees with $\p{h}(x)$ for each $x\in R_h$, we have
\[
\prob{x}{f_I(x)\neq P(x)\ |\ x\in R_h} = \prob{x}{f_I(x)\neq \p{h}(x)\ |\ x\in R_h} \leq \varepsilon_0.
\]

Hence, we obtain
\begin{align*}
\prob{x}{f_I(x)\neq P(x)\ |\ x\in R_h\setminus \bigcup_{h'< h}R_{h'}} &\leq \frac{\prob{x}{f_I(x)\neq P(x)\ |\ x\in R_h}}{\prob{x}{x\in R_h\setminus \bigcup_{h'< h}R_{h'}\ |\ x\in R_h}}\\
&= 2^{h-1}\prob{x}{f_I(x)\neq P(x)\ |\ x\in R_h} \leq 2^{h-1}\varepsilon_0.
\end{align*}

Plugging the above into (\ref{eq:fPdist}), we get
\begin{align*}
\prob{x}{f_I(x) \neq P(x)} &\leq \left(\sum_{h=1}^{\ell} 2^{h-1}\varepsilon_0 \cdot \frac{1}{2^{h}}\right) + \frac{1}{2^{\ell}}\\
&\leq \frac{\varepsilon_0\ell}{2} + \frac{1}{2^{\ell}} < \varepsilon_1
\end{align*}
where the final inequality follows from our choice of $\varepsilon_0$ and $\ell$ (see  (\ref{eq:choiceelleps0})). This is a contradiction to our assumption on $\delta_d(f_I)$, which concludes the proof of Lemma~\ref{lem:LDL} assuming Claim~\ref{clm:global}.

\subsubsection{Proof of Claim~\ref{clm:global}}

We now prove Claim~\ref{clm:global}. The proof is a case analysis based on whether $G_{\text{bad}}$ has a large matching or a large star. For any $h\in [\ell+1]$ and any polynomial $Q\in \Fnd{I}{d}$, we denote $Q|_h$ the polynomial obtained by identifying the variables $X_{i_{h}}$ and $X_{j_h}$. We want to define a polynomial $P$ such that for each $h\in [\ell+1]$, we have
\begin{equation}
\label{eq:PPh}
P|_h = P^{(h)}.
\end{equation}

As in~\cite{BKSSZ}, the crucial observation that will help us find a $P$ as above is the following. Fix any distinct $h,h'$ and consider $\p{h}|_{h'}$ and $\p{h'}|_h$. Note that these polynomials are both naturally defined on the set of inputs $R_{h,h'} := R_h\cap R_{h'}$. However, since $f_I$ is $\varepsilon_1$-close to $\p{h}$ and $\p{h'}$ on $R_h$ and $R_{h'}$ respectively, we see that
\begin{align*}
\prob{x\in R_{h,h'}}{\p{h}|_{h'}(x) \neq \p{h'}|_h(x)} &= \prob{x\in R_{h,h'}}{\p{h}(x) \neq \p{h'}(x)}\\ &\leq \prob{x\in R_{h,h'}}{\p{h}(x) \neq f(x)} + \prob{x\in R_{h,h'}}{\p{h'}(x) \neq f(x)}\\
&\leq 2\varepsilon_0 + 2\varepsilon_0 \leq 4\varepsilon_0 < \frac{1}{2^d},
\end{align*}
where for the second inequality we have used the fact that
\[
\prob{x\in R_{h,h'}}{\p{h}(x) \neq f(x)} \leq \frac{\prob{x\in R_{h}}{\p{h}(x) \neq f(x)}}{\prob{x\in R_{h}}{x\in R_{h'}}} \leq \frac{\varepsilon_0}{1/2} = 2\varepsilon_0.
\]

Since any pair of distinct polynomials of degree $d$ disagree on at least a $(1/2^d)$ fraction of inputs in $R_{h,h'}$, we see that $P^{(h)}|_{h'} = P^{(h')}|_h$ as polynomials. We record this fact below.

\begin{claim}
\label{clm:local}
For any distinct $h,h'\in [\ell+1]$, $P^{(h)}|_{h'} = P^{(h')}|_{h}$.
\end{claim}

\paragraph{The Matching case of Claim~\ref{clm:global}.}

Let $(X_{i_1},X_{j_1},0),\ldots,(X_{i_{\ell+1}},Y_{i_{\ell+1}},0)$ be the set of bad triples that give rise to the distinct edges of the matching in $G_{\text{bad}}$. By renaming variables we assume that $I = [r]$ and the bad triples are all of the form $(X_1,X_2,0),\ldots,(X_{2\ell+1},X_{2\ell+2},0)$.

Assume that for each $h\in[\ell+1]$,
\[
P^{(h)}(X) = \sum_{S\subseteq I\setminus \{2h\}: |S|\leq d}\alpha_S^{(h)} X^S
\]
where $X^S = \prod_{i\in S}X_i$ (note that $P^{(h)}\in \Fnd{I\setminus \{2h\}}{d}$ and hence does not involve $X_{2h}$). For any $h$, if $|S|> d$ or $S\ni 2h$, we define $\alpha_S^{(h)} = 0$.

Note that we have for any distinct $i,j\in [\ell+1]$
\begin{equation}
\label{eq:Pij1}
P^{(i)}(X)|_j = \sum_{S\cap \{2j-1,2j\} = \emptyset} \alpha_S^{(i)} X^S + \sum_{S\cap \{2j-1,2j\} = \emptyset} (\alpha_{S\cup\{2j-1\}}^{(i)}+\alpha_{S\cup\{2j\}}^{(i)}+\alpha_{S\cup\{2j-1,2j\}}^{(i)})X^{S\cup\{2j-1\}}.
\end{equation}

In particular, Claim~\ref{clm:local} implies the following for $S\subseteq I$ such that $|S|\leq d$ and $i,j$ distinct such that $S\cap \{2i-1,2i,2j-1,2j\} = \emptyset$,
\begin{align}
\alpha_S^{(i)} &= \alpha_S^{(j)}\label{eq:localcons1}\\
\alpha_{S\cup \{2i-1\}}^{(i)} &= \alpha^{(j)}_{S\cup \{2i-1\}} + \alpha^{(j)}_{S\cup \{2i\}} + \alpha^{(j)}_{S\cup \{2i-1,2i\}}\label{eq:localcons2}
\end{align}
Let $\alpha_S^{(i)}|_j$ denote the coefficient of $X^{S}$ in $P^{(i)}|_j$.

We define the polynomial
\[
P(X) = \sum_{S\subseteq I: |S|\leq d} \alpha_S X^S
\]
as follows. For each $S\in \binom{I}{\leq d}$, set $\alpha_S = \alpha_S^{(j)}$ for any $S$ such that $S\cap \{2j-1,2j\} = \emptyset$: since $|S|\leq d \leq \ell$, there is at least one such $j\in [\ell+1]$. By (\ref{eq:localcons1}), we see that any choice of $j$ as above yields the same coefficient $\alpha_S$.

Note that
\begin{equation}
\label{eq:P|i}
P|_i = \sum_{S\cap \{2i-1,2i\}=\emptyset} \alpha_S X^S + \sum_{S\cap \{2i-1,2i\}=\emptyset} (\alpha_{S\cup\{2i-1\}}+\alpha_{S\cup\{2i\}}+\alpha_{S\cup\{2i-1,2i\}}) X^{S\cup \{2i-1\}}.
\end{equation}
Let $\alpha_S|_j$ denote the coefficient of $X^{S}$ in $P|_j$.

Now we show that $P|_i = P^{(i)}$ for each choice of $i\in [\ell+1]$ by comparing coefficients of monomials and showing that $\alpha_S|_i = \alpha_S^{(i)}$ for each $S$ such that $|S|\leq d$. That will conclude the proof of the matching case of Claim~\ref{clm:global}. Fix any $S$ such that $|S|\leq d$. We consider three cases.
\begin{itemize}
\item $S\ni 2i$: In this case, $\alpha_S|_i= \alpha_S^{(i)} = 0$ and hence we are done.
\item $S\cap \{2i-1,2i\}=\emptyset$: In this case, we have $\alpha_S|_i = \alpha_S^{(i)}$ by definition and hence we are done.
\item $S\cap \{2i-1,2i\}=2i-1$: In this case, let $T = S\setminus \{2i-1\}$ and fix $j\in [\ell+1]$ such that $j\neq i$ and $T\cap \{2j-1,2j\} = \emptyset$. We see that
\begin{align*}
\alpha_{S}|_i &= \alpha_{T\cup\{2i-1\}}+\alpha_{T\cup\{2i\}}+\alpha_{T\cup\{2i-1,2i\}}\qquad &\text{(by (\ref{eq:P|i}))}\\
&= \alpha_{T\cup\{2i-1\}}^{(j)}+\alpha_{T\cup\{2i\}}^{(j)}+\alpha_{T\cup\{2i-1,2i\}}^{(j)}\qquad &\text{(by definition of $P$)}\\
&= \alpha_{T\cup \{2i-1\}}^{(i)}\qquad &\text{(by (\ref{eq:localcons2}))}\\
&= \alpha_S^{(i)}. &
\end{align*}

\end{itemize}

\paragraph{The Star case of Claim~\ref{clm:global}.} We proceed as in the matching case, except that the definition of $P$ will be somewhat more involved. By renaming variables we assume that $I = [r]$ and that the bad triples are all of the form $(X_1,X_r,0),(X_2,X_r,0),\ldots,(X_{\ell+1},X_r,0)$.

Assume that for each $h\in[\ell+1]$,
\[
P^{(h)}(X) = \sum_{S\subseteq [r-1]: |S|\leq d}\alpha_S^{(h)} X^S
\]
where $X^S = \prod_{i\in S}X_i$ (note that $P^{(h)}\in \Fnd{I\setminus \{r\}}{d}$ and hence does not involve $X_{r}$). For any $h$, if $|S|> d$ or $S\ni r$, we define $\alpha_S^{(h)} = 0$.

For any distinct $i,j\in [\ell+1]$ with $i < j$, we assume that $P^{(i)}|_j$ and $P^{(j)}|_i$ are obtained by replacing $X_j$ with $X_i$. We thus have
\begin{align}
P^{(i)}(X)|_j = \sum_{S\cap \{i,j\} = \emptyset} \alpha_S^{(i)} X^S + \sum_{S\cap \{i,j\} = \emptyset} (\alpha_{S\cup\{i\}}^{(i)}+\alpha_{S\cup\{j\}}^{(i)}+\alpha_{S\cup\{i,j\}}^{(i)})X^{S\cup\{i\}}.\label{eq:Pij2}\\
P^{(j)}(X)|_i = \sum_{S\cap \{i,j\} = \emptyset} \alpha_S^{(j)} X^S + \sum_{S\cap \{i,j\} = \emptyset} (\alpha_{S\cup\{i\}}^{(j)}+\alpha_{S\cup\{j\}}^{(j)}+\alpha_{S\cup\{i,j\}}^{(j)})X^{S\cup\{i\}}.\label{eq:Pji}
\end{align}

Using Claim~\ref{clm:local} and comparing coefficients of $P^{(i)}|_j$ and $P^{(j)}|_i$, we get for $i\neq j$ and $S$ such that $S\cap \{i,j\} = \emptyset$,
\begin{align}
\alpha_S^{(i)} &= \alpha_S^{(j)}\label{eq:localcons3}\\
\alpha_{S\cup\{i\}}^{(i)}+\alpha_{S\cup\{j\}}^{(i)}+\alpha_{S\cup\{i,j\}}^{(i)} &= \alpha_{S\cup\{i\}}^{(j)}+\alpha_{S\cup\{j\}}^{(j)}+\alpha_{S\cup\{i,j\}}^{(j)}\label{eq:localcons4}
\end{align}

We now define the polynomial
\[
P(X) = \sum_{S\subseteq I \setminus \{r\}: |S|\leq d} \beta_S X^S
+ \sum_{S\subseteq I: S \ni r, |S|\leq d} \gamma_S X^S
\]
as follows. 
\begin{itemize}
\item For $S\not\ni r$, we define $\beta_S$ to be $\alpha_S^{(i)}$ for any $i\in [\ell+1]$ such that $i\not\in S$. Since $|S|\leq d < \ell+1$ there is such an $i$. Note that by (\ref{eq:localcons3}), the choice of $i$ is immaterial.

\item For $S\ni r$, we let $T = S\setminus \{r\}$. Note that $|T| < d$. We define $\gamma_{T \cup \{r\}}$ by downward induction on
$|T|$ as follows:
$$\gamma_{T\cup \{r\}} \triangleq \alpha^{(i)}_{T\cup \{i\}} - \beta_{T\cup \{i\}} - \gamma_{T\cup \{i,r\}} \mbox{ for any fixed }i\in [\ell+1]\setminus T$$
where we assume that $\gamma_{T\cup \{r\}}=0$ for $|T|\geq d$.
\end{itemize}

We will show first by downward induction on $|T|$ that these coefficients are independent of the choice of $i\in [\ell+1]\setminus T$. Fix $i,j \in [\ell+1]\setminus T$. In the base case $|T| = d-1$, we have 
\begin{align*}
\alpha^{(i)}_{T\cup \{i\}} - \beta_{T\cup \{i\}} &= \alpha^{(i)}_{T\cup \{i\}} - \alpha^{(j)}_{T\cup \{i\}}&\text{(by definition of $\beta_{T\cup \{i\}}$)}\\
&= \alpha^{(j)}_{T\cup \{j\}} - \alpha^{(i)}_{T\cup \{j\}}&\text{(by (\ref{eq:localcons3}) and $\alpha_{T\cup \{i,j\}} = 0$)}\\
&= \alpha^{(j)}_{T\cup \{j\}} - \beta_{T\cup \{j\}}&\text{(by definition of $\beta_{T\cup \{j\}}$)} 
\end{align*}

When $|T| < d-1$, we have

\begin{eqnarray*}
\lefteqn{\alpha^{(i)}_{T\cup \{i\}} - \beta_{T\cup \{i\}} - \gamma_{T\cup
\{i,r\}}} \\
& = &
\alpha^{(i)}_{T\cup \{i\}} - \alpha^{(j)}_{T\cup \{i\}} - \gamma_{T\cup \{i,r\}}\\
& = &
\alpha^{(i)}_{T\cup \{i\}} - \alpha^{(j)}_{T\cup \{i\}} -
\left(
\alpha^{(j)}_{T\cup \{i,j\}} - \beta_{T\cup \{i,j\}} - \gamma_{T\cup \{i,j,r\}}\right)
\mbox{~(Applying definition of $\gamma_{T \cup \{i,r\}}$ and induction)}\\
& = &
\alpha^{(j)}_{T\cup \{j\}} - \alpha^{(i)}_{T\cup \{j\}} -
\left(
\alpha^{(i)}_{T\cup \{i,j\}} - \beta_{T\cup \{i,j\}} - \gamma_{T\cup \{i,j,r\}}\right)
\mbox{~~~(From (\ref{eq:localcons4}) and rearranging terms.)}\\
& = &
\alpha^{(j)}_{T\cup \{j\}} - \alpha^{(i)}_{T\cup \{j\}} -
\gamma_{T\cup \{j,r\}}
\mbox{~~~~~~(Applying definition of $\gamma_{T \cup \{j,r\}}$ and induction)}\\
& = &
\alpha^{(j)}_{T\cup \{j\}} - \beta_{T\cup \{j\}} -
\gamma_{T\cup \{j,r\}}
\end{eqnarray*}

We now conclude by showing that the restriction of $P$ obtained by
replacing $X_r$ by $X_i$ equals the polynomial $P^{(i)}$.
Let $P|_i$ denote the restriction of $P$ obtained by replacing $X_r$ by $X_i$
and let $\alpha_S|_i$ be its coefficients. Note that
\begin{eqnarray*}
P|_i & = &
\sum_{S\not\ni i:|S|\leq d} \beta_S X^S + \sum_{T\not\ni i: |T| < d }
(\beta_{T\cup\{i\}}+\gamma_{T\cup\{r\}}+\gamma_{T\cup\{i,r\}})X^{T\cup\{i\}}.\\
& = &
\sum_{S\not\ni i:|S|\leq d} \alpha^{(i)}_S X^S +
\sum_{T\not\ni i: |T| < d }
\alpha^{(i)}_{T \cup \{i\}} X^{T\cup \{i\}} \mbox{~~~(By definition of
$\beta_S$ and $\gamma_{T \cup \{r\}}$)}\\
& = & P^{(i)}.
\end{eqnarray*}
This concludes the proof for the star case.

\section{Impossibility of local decoding when $\charF(\F)$ is large}
\label{sec:imposs}

In this section, we prove Theorem~\ref{thm:imposs-large-char} which is a more detailed version of Theorem~\ref{thm:imposs-large-char-intro}.  Again we remind the reader that an overview may be found in Section~\ref{ssec:overview}. 

Let $n$ be a growing parameter and $\F$ a field of characteristic $0$ or positive characteristic greater than $n^2$. For the results in this section, it will be easier to deal with the domain $\{-1,1\}^n$ rather than $\{0,1\}^n$. Since there a natural invertible linear map that maps $\{0,1\}$ to $\{-1,1\}$ (i.e. $a\mapsto 1-2a$), this change of input space is without loss of generality.


\subsection{Local linear spans of balanced vectors}
\label{ssec:imposs-tech}

Let $u\in \F^n$ and $U\subseteq \F^n$. For any integer $t\in \mathbb{N}$, we say that $u$ is in the $t$-span of $U$ if it can be written as a linear combination of at most $t$ elements of $U$. For $x\in \{-1,1\}^n$, we use $|x|$ to denote the sum of the entries of $x$ over $\mathbb{Z}$. In this section, we wish to show that if the vector $1^n$ is in the $t$-span of
balanced vectors, i.e., vectors $x$ with $|x| \leq n/s$ then $t$ is must be growing as a function of $s$. 

As explained earlier we first establish a bound on the size of the solutions of linear equations in systems over $\mathbb{Q}$ with few variables or few constraints.  
This fact is well-known, but we prove it here for completeness.


\begin{lemma}
\label{lem:Cramer}
Let $r,s\in \mathbb{N}$ and let $t = \min\{r,s\}$. Let $Mx = u$ be a system of linear equations with $M\in \{-1,0,1\}^{r\times s}$ and $u\in \{-1,0,1\}^r$. 
\begin{itemize}
\item
If $\F$ is a field of characteristic and the system  has a solution in $\F^s$, then there exist integers
$a_1,\ldots,a_s,b \in \mathbb{Z}$ with $|a_i|,|b| \leq t!$ such that $x_i = a_i/b$ is a solution to $Mx=u$. In particular, there is a solution in $\mathbb{Q}^s$.
\item If $\F$ is a field of characteristic $p$ and if the system  has a solution in $\F^s$,
then there exist integers
$a_1,\ldots,a_s,b \in \mathbb{Z}$ with $|a_i|,|b| \leq t!$ such that $x_i = a_i / b \pmod{p}$ is a solution to $Mx=u$. In particular, there is a solution in $\F_p^s$.
\end{itemize}
\end{lemma}

\begin{proof}
Note that we can assume that $M$ has full column rank. This is because $Mx = u$ has a solution iff $\tilde{M}\tilde{x} = u$ has a solution where $\tilde{M}$ is the submatrix of $M$ obtained by a keeping a maximal set of linearly independent columns of $M$. When the columns are linearly independent, we have $s$ is at most $r$ and hence $t=\min\{r,s\} = s$. 

We start with the zero characteristic case.
Let $M'$ be an invertible $s\times s$ submatrix of $M$ containing the set of $s$ linearly independent rows of $M$ and let $u'\in \F^s$ be the vector corresponding to these rows. Note that the solution $x$ is uniquely determined by $M'x = u'$. We now apply Cramer's rule to see that the solution is given by 
\[
x_i = \frac{\det(M'_i)}{\det(M')}
\]
for  $i\in [s]$, where $M'_i$ is the $s\times s$ matrix obtained by replacing the $i$th column of $M'$ bu $u'$. Since $M'$ and $M'_i$ are  matrices with entries in $\{-1,0,1\}$,   we have $\det(M')\in \mathbb{Z}$ with $|\det(M')|\leq s!$ for each $i\in [s]$ and similarly for $\det(M'_i)$. Therefore, we have the claim with $a_i = \det(M'_i)$ and $b = \det(M')$. 

The characteristic $p$ case is similar with only differnece being the solution now is given by $x_i = a_i/b \pmod{p}$.
\end{proof}

We now turn to the main technical lemma of this section showing that $1^n$ is not in linear span of a small number of nearly balanced elements of $\{-1,1\}^n$.

\begin{lemma}
\label{lem:smallspan}
Let $n,s = s(n)\in \mathbb{N}$ with $s(n)\leq n$. Let $S = \{x\in \{-1,1\}^n\ |\ |x| \leq n/s)\}$. Then $x^0 = 1^n$ is not in the $t$-span of $S$ unless $t\geq \log s/\log \log s$ provided $\F$ is field of zero characteristic or of characteristic $p \geq 2n^2$.
\end{lemma}

\begin{proof}
We first consider the case when $\F$ is of zero characteristic. Note that in this case $\mathbb{Q} \subseteq \F$. Suppose 
 $x^0\in \mathrm{Span}\{x^1,\ldots,x^t\}$ with $x^0 = \sum_{i=1}^n c_i x^i$. Note that the $c_i$'s are expressible as the solution to a linear system whose $M z = u$ where $M$ and $u$ have entries in $\{-1,0,1\}$ and $M$ is a $n \times t$ matrix. By Lemma~\ref{lem:Cramer} we have that $c_i \in \mathbb{Q}$ with $|c_i| \leq t!$ (more specifically we have $c_i = a_i/b$
 with $|a_i| \leq t!$ and this implies $|c_i| \leq t!$). We thus have 
$$n
= \left| \sum_{j=1}^n x^0_j \right|
= \left| \sum_{i=1}^t c_i \sum_{j=1}^n x^i_j \right|
\leq \sum_{i=1}^t \left| c_i \right | \cdot \left | \sum_{j=1}^n x^i_j \right|
\leq \sum_{i=1}^t (t!) \cdot (n/s) \leq (t+1)!\cdot (n/s).$$
We thus conclude that $(t+1)! \geq s$ and thus $t \geq \log s/\log \log s$.

In the case of finite field $\F$, we proceed as above and let $x^0 = \sum_{i=1}^t c_i x^i$. By Lemma~\ref{lem:Cramer} we have that there are integers $a_i,b$ with $|a_i|,|b| \leq t!$ such that $c_i = a_i/b \pmod{p}$ is a solution to $x^0 = \sum_{i=1}^t c_i x^i$.
Now consider $b\cdot n$ and we get $b \cdot n = \sum_{i=1}^t a_i \sum_{j=1}^n x^i_j \pmod{p}$. We now show that this implies $(t+1)! \geq \min\{p/(2n),s\} = s$ (where the equality follows from $p \geq  2n^2$ and $s \leq n$). Assume $(t+1)! \leq p/(2n)$. Then
we have $n \leq |b \cdot n| \leq t! \cdot n < p/2$ over the integers, and 
$\left|\sum_{i=1}^t a_i \sum_{j=1}^n x^i_j\right| \leq (t+1)! (n/s) < p/2$ also over the integers. We again conclude that $n \leq (t+1)! (n/s)$ and so $(t+1)! \geq s$ as claimed.
The lemma follows. 
\end{proof}

\subsection{Proof of Theorem~\ref{thm:imposs-large-char-intro}}
\label{ssec:imposs-proof}

We now state and prove Theorem~\ref{thm:imposs-large-char} which immediately implies Theorem~\ref{thm:imposs-large-char-intro}.

\begin{theorem}
\label{thm:imposs-large-char}
Let $n\in \mathbb{N}$ be a growing parameter and $\varepsilon\in (0,1)$ such that $\varepsilon \geq 2\exp(-n/2s^2)$ for some $s\in \mathbb{N}$ with $100\leq s \leq \sqrt{n}/100$. Let $\F$ be any field such that either $\charF(\F) = 0$ or $\charF(\F)\geq n^2$. Then any adaptive $(\varepsilon,q)$-local decoder for $\Fnd{n}{1}$ that corrects an $\varepsilon$ fraction of errors must satisfy $q = \Omega(\log s/\log \log s)$. 
\end{theorem}

\begin{proof}
The proof of the theorem will use the minimax principle. Specifically, we design a ``hard'' probability distribution $\mc{D}$ over functions that are $\varepsilon$-close to $\Fnd{n}{1}$ such that any \emph{deterministic} decoder that  decodes the value of a random function (chosen according to $\mc{D}$) at the point $1^n$ while making very few queries will fail to decode the value with probability at least $1/4$.

We start with the case of positive characteristic which is somewhat simpler to describe. Let $\charF(\F) = p > n^2$. We define the hard distribution $\mc{D}$ as follows. Let
\[
E = \{x\in \{-1,1\}^n\ |\ |\sum_i x_i| \geq 2n/s\},
\]
so that, by the Chernoff bound, (see, e.g.~\cite{DP}) $|E|\leq \varepsilon2^n.$ 
Let $S = \{-1,1\}^n\setminus E.$ 

We now sample a random function $f\sim \mc{D}$ as follows:
\begin{itemize}
\item Choose $a_1,\ldots,a_n\in \F_p\subseteq \F$ uniformly at random independently. Let $\ell(X_1,\ldots,X_n) = \sum_i a_i X_i\in \Fnd{n}{1}.$
\item Let $f(x) = 0$ if $x\in E$ and $f(x) = \ell(x)$ if $x\in S.$
\end{itemize}
Since $f(x) = \ell(x)$ for $x\not\in E,$ we have $\delta(f,\ell)\leq \varepsilon$. In particular $\delta_1(f)\leq \varepsilon.$

Let $\mc{A}$ be any deterministic decoding algorithm for decoding $f(1^n).$ Assume that the worst case number of queries $t$ made by $\mc{A}$ satisfies $t < \log s/\log \log s.$ W.l.o.g. we assume that $\mc{A}$ always makes exactly $t$ queries and also that none of these queries are made to inputs $x\in E$ (since at these points $f(x)$ is known to take the value $0$). Additionally, we may assume that these queries correspond to linearly independent inputs since if a query point $x$ is a linear combination of previous queries, then $\ell(x) = \ip{a}{x}$ can be determined from the answers to previous queries. 

Let $x^1,\ldots,x^t$ be the (adaptive) queries made by $\mc{A}$ on the random function $f$. After these queries are made, the algorithm has $\ell(x^i) = \ip{a}{x^i}$ for each $i\in [t]$, where $a = (a_1,\ldots,a_n)$. However, by Lemma~\ref{lem:smallspan}, we know that $1^n$ is not in the $t$-span of the inputs in $S$ and hence, given the values $\ell(x^1),\ldots,\ell(x^t)$, $\ell(1^n) = \sum_i a_i$ is still distributed uniformly over $\F_p$. Hence, the probability that the algorithm outputs $\ell(1^n)$ correctly is at most $1/p < 3/4.$ 

Now consider the case when $\charF(\F) = 0.$ We define our hard distribution $\mc{D}$ exactly as above except that the coefficients $a_1,\ldots,a_n$ are chosen i.u.a.r. from $\{-N,\ldots,N\}$ where $N = n^{\lceil\log s/\log\log s \rceil}$.

Let $\mc{A}$ be any deterministic decoding algorithm for decoding $f(1^n)$ as above. Again, we assume that $\mc{A}$ always makes $t\leq \log s/\log \log s$ many queries corresponding to linearly independent inputs, and also that none of these queries are made to inputs $x\in E.$ 

Let $A \subseteq \{-N,\ldots,N\}^n$ be the set of coefficients of linear polynomials $\ell$ such that $\mc{A}$ is able to decode $\ell(1^n) = \sum_i a_i$ correctly. 

To bound the size of $|A|$, we use an encoding argument. Considre any $(a_1,\ldots,a_n)\in A$ and let $\ell(X) = \sum_i a_i X_i.$ Let $x^1,\ldots,x^t$ be the queries made on input $\ell$. Given $\ip{a}{x^i}$ for $i\in [t],$ the algorithm determines $\sum_i a_i = \ip{a}{1^n}$. Hence, at this point the algorithm has $\ip{a}{x}$ for $x\in I' = \{x^1,\ldots,x^t, 1^n\}$. Note that $I'$ is a set of dimension $t+1$ since by Lemma~\ref{lem:smallspan}, $1^n$ is not in the $t$-span of $S$. We can thus a subset $I'' = \{e^{i_1},\ldots,e^{i_{n-t-1}}\}$ of the set of standard basis vectors $\{e^1,\ldots,e^n\}$ of size $n-t-1$ so that $I = I'\cup I''$ is a basis for $\F^n.$ 

Define an encoding function 
\[
\mc{E}:A\rightarrow \{-Nn,\ldots,Nn\}^t \times \{-N,\ldots,N\}^{n-t-1}
\]
as follows. For each $x\in A$, we choose $I$ as above and set
\[
\mc{E}(a) = (\ip{a}{x^1},\ldots,\ip{a}{x^t},\ip{a}{e^1},\ldots,\ip{a}{e^{i_{n-t-1}}}).
\]
Note that each $\ip{a}{x^j}\in \{-Nn,\ldots,Nn\}$ since $a\in \{-N,\ldots,N\}^n$ and $x^j\in\{-1,1\}^n$.  

We claim that $\mc{E}$ is $1$-$1$. This is because, on being given $\mc{E}(a)$ as above, we can determine $\ip{a}{x}$ for each $x\in I$ by the following argument: $\mc{E}(a)$ directly gives us $\ip{a}{x}$ for each $x\in I\setminus \{1^n\}$ and by construction of $x^1,\ldots,x^t$, we know that $\ip{a}{x^1},\ldots,\ip{a}{x^t}$ determines the value of $\ip{a}{1^n}.$ Thus, we have $\ip{a}{x}$ for each $x\in I$ and as $I$ is a basis for $\F^n$, we can obtain $a\in \F^n$ as well. 

Since $\mc{E}$ is $1$-$1$, we see that 
\[
|A| \leq (2Nn+1)^t\cdot (2N+1)^{n-t-1}\leq (2N+1)^{n-1}\cdot n^t \leq (2N+1)^{n-1} \cdot N \leq (2N+1)^{n} \cdot \frac{3}{4}.
\]
which implies that the relative size of $A$ inside $\{-N,\ldots,N\}^n$ is at most $3/4$. This concludes the proof.
\end{proof}

\section{Local decoding when $\charF(\F)$ is small}
\label{sec:ldsmall}

In this section, we give a local decoder over fields of small characteristic. An overview of this construction may be found in Section~\ref{ssec:overview}.

Let $p$ be a prime of constant size and let $\F$ be any (possibly infinite) field of characteristic $p$. Let $d$ be the degree parameter and $k$ be the smallest power of $p$ that is strictly greater than $d$. Note that $k\leq pd.$ 
We show that the space $\Fnd{n}{d}$ has a $(1/(4\cdot \binom{2k}{k}),\binom{2k}{k})$-local decoder, hence proving Theorem~\ref{thm:decoding-small-char-intro}. 

The main technical tool we use is a suitable linear relation on the space $\Fnd{2k}{d},$ which we describe now. 
We say that a set $S\subseteq \{0,1\}^{2k}$ is \emph{useful} if for every polynomial $G\in \Fnd{2k}{d}$, $G(0^{2k})$ is determined by the restriction of the function $G$ to the inputs in $S$. 
Let $B\subseteq \{0,1\}^{2k}$ denote the set of all balanced inputs (i.e. inputs of Hamming weight exactly $k$). 

\begin{lemma}
\label{lem:useful}
Fix $d,k$ as above. Then the set $B\subseteq \{0,1\}^{2k}$ of balanced inputs is useful.  
\end{lemma}

The proof of the above lemma will use Lucas' theorem, which we recall below.

\begin{theorem}[Lucas' theorem]
\label{thm:Lucas}
Let $p$ be any prime and $a,b\in \mathbb{N}$. Let $a_1,\ldots,a_\ell\in \{0,\ldots, p-1\}$ and $b_1,\ldots,b_\ell\in \{0,\ldots,p-1\}$ be the digits in the $p$-ary expansion of $a$ and $b$,  i.e.,  $a = \sum_{j\in [\ell]}a_j p^{j-1}$ and $b = \sum_{j\in [\ell]} b_j p^{j-1}$. Then, we have
\[
\binom{a}{b} \equiv \prod_{i\leq \ell} \binom{a_i}{b_i} \pmod{p}
\]
where $\binom{a_i}{b_i}$ is defined to be $0$ if $a_i < b_i.$
\end{theorem}

\begin{corollary}
\label{cor:Lucas}
For $i\in \{0,\ldots,d\}$, we have $\binom{d+k-i}{k-i}\not\equiv 0\pmod{p}$ if and only if $i=0$. 
\end{corollary}

\begin{proof}
Note that by Lucas' theorem (Theorem~\ref{thm:Lucas}), $\binom{a}{b} \equiv 0 \pmod{p}$ if and only if there are digits $a_j,b_j$ in the $p$-ary expansions of $a$ and $b$ respectively with $a_j < b_j.$

Consider first the case when $i=0$. Let $a = d + k$ and $b=k$. Let
\begin{equation}
\label{eq:ab}
a = \sum_{j=1}^{\ell} a_j\cdot p^{j-1}\qquad\qquad b = \sum_{j=1}^{\ell} b_j\cdot p^{j-1}
\end{equation}
where $a_j,b_j\in \{0,\ldots, p-1\}$ and $k=p^{\ell-1}$. Then, we have $b_{j} = 0$ for each $j < \ell$ and $b_{\ell} = a_{\ell} = 1$. Hence by Lucas' theorem, we have $\binom{a}{b}\neq 0\pmod{p}.$

Now consider the case when $i\in [d].$ Let $a = d+k-i$ and $b=k-i$. Again write $a,b$ as in (\ref{eq:ab}) with $k = p^{\ell-1}$. In this case, we have $a_\ell = 1$ but $b_\ell = 0$, the latter due to the fact that $b < k.$ Hence if we consider $a' = \sum_{j \in [\ell-1]}a_j p^{j-1}$ and $b' = \sum_{j \in [\ell-1]}b_j p^{j-1}$, we get $a' = d-i < b' = k-i$. Therefore, there must exist $j \in [\ell-1]$ such that $a_j < b_j$. From Lucas' theorem, it now follows that $\binom{a}{b} \equiv 0\pmod{p}.$
\end{proof}

\begin{proof}[Proof of Lemma~\ref{lem:useful}]
Fix any $G\in \Fnd{2k}{d}$. Assume that
\[
G(Y_1,\ldots,Y_{2k}) = \sum_{I\subseteq [2k]: |I|\leq d} \alpha_I Y^{I}
\]
where $Y^I$ denotes $\prod_{i\in I} Y_i$.

Let $B'$ denote all those inputs in $B$ where the last $k-d$ bits are set to $0$. We will compute the sum of $G$ on inputs from $B'$. But let us first consider a monomial $Y^I$ and see what its sum over $y \in B'$ looks like. The monomial evaluates to $1$ on $y \in B'$ if $y_i = 1$ for every $i \in I$, and evaluates to $0$ otherwise. There are exactly $\binom{d+k - |I|}{k-|I|}$ choices of $y\in B'$ that satisfy $y_i = 1$ for every $i \in I$. Thus summing over $y \in B'$ we get $\sum_{y\in B'} y^I = \binom{d+k-|I|}{k - |I|}$. Summing over all monomials we get:
\begin{align}
\sum_{y\in B'} G(y) &=  \sum_{I\subseteq [2k]: |I|\leq d} \alpha_I \cdot \sum_{y\in B'} Y^I\notag\\
&= \sum_{I\subseteq [2k]: |I|\leq d} \alpha_I \cdot \binom{d+k-|I|}{k-|I|}\label{eq:Lucas}
\end{align}

By Corollary~\ref{cor:Lucas}, it follows that for $i\in \{0,\ldots,d\}$, we have
\[
\binom{d+k-i}{k-i} \not\equiv 0\pmod{p}
\]
if and only if $i=0$ and so $\sum_{y \in B'} G(y) = \binom{d+k}{k}\cdot \alpha_{\emptyset}$.  Let $c = \binom{d+k}{k} \pmod{p}$. We have
$c \in \F_p^* \subseteq \F^*$ and in particular $c$ is invertible in $\F$, and $\sum_{y\in B'}G(y) = c\cdot \alpha_{\emptyset} =  c\cdot G(0^{2k})$. 
Hence, we get $G(0^{2k}) = c^{-1}\cdot \sum_{y\in B'} G(y).$ Therefore, $G(0^{2k})$ is determined by the restriction of $G$ to $B'$ and hence also by its restriction to $B$.
\end{proof}

We now show that $\Fnd{n}{d}$ has a $(1/(4\cdot \binom{2k}{k}),\binom{2k}{k})$-local decoder. 

\paragraph{The decoder.} We now give the formal description of the decoder. Let the decoder be given oracle access to $f$ with the promise that $f$ is $1/(4\cdot \binom{2k}{k})$-close to some $F\in \Fnd{n}{d}$. Let the input to the decoder be $x\in \{0,1\}^n$. The problem is to find $F(x)$.

We describe the decoder below:

\paragraph{Decoder $D_k^f(x)$.}
\begin{itemize}
\item Partition $[n]$ into $2k$ parts by choosing a \emph{uniformly} random map $h:[n]\rightarrow [2k]$. I.e. each $h(j)$ is chosen i.u.a.r. from $[2k].$
\item For $i\in [2k]$ and $j\in [n]$ such that $h(j) = i$, identify $X_j$ with $Y_i \oplus x_j$.
\item Let $g(Y_1,\ldots, Y_{2k})$ and $G(Y_1,\ldots,Y_{2k})$ be the restrictions of $f$ and $F$ respectively. Assuming $g|_B=G|_B$, query $g$ at all inputs in $B$ and decode $G(0^{2k})$ from $G|_B$. Output the value decoded.
\end{itemize}

The main theorem of this section is the following. Note that this implies Theorem~\ref{thm:decoding-small-char-intro}.

\begin{theorem}
\label{thm:decoding-small-char}
Let $\F$ be a field of characteristic $p$. For integer $d \geq 0$, let $k$ be the smallest power of $p$ greater than $d$. Then
the decoder $D_k$ is a $(1/(4\cdot \binom{2k}{k}),\binom{2k}{k})$-local decoder for $\Fnd{n}{d;\F}.$
\end{theorem}

\begin{proof}
The bound on the query complexity of the decoder is clear from the description of $D_k$. So we only need to argue that the decoder outputs the value of $F(x)$ correctly with probability at least $3/4$.

The crucial observation is that for each fixed $y\in B$, querying $g(y)$ amounts to querying $f$ at a \emph{uniformly} random point $z\in \{0,1\}^n$, where the randomness comes from the choice of $h$. This is because for each $j\in [n]$, we have
\[
z_j = y_{h(i)} \oplus x_j
\]
where $h:[n]\rightarrow [2k]$ is a uniformly random function. Since $y$ is \emph{balanced}, each $y_{h(i)}$ is a uniformly random bit. Hence we see that $z\in \{0,1\}^n$ is distributed uniformly over $\{0,1\}^n.$

Thus, if $\delta(f,F) \leq 1/(4\cdot \binom{2k}{k})$, with probability at least $3/4$, all the random queries made lie outside the error set $E = \{z\in \{0,1\}^n\ |\ f(z)\neq F(z)\}$ and in this case, the decoder is able to access the function $G|_B$ at each input $y\in B.$ By Lemma~\ref{lem:useful}, this allows the decoder to determine $G(0^{2k}).$ Noting that the image of $0^{2k}$ in $\{0,1\}^n$ is exactly $x$, we thus see that the decoder outputs $F(x)$ correctly. 
\end{proof}

%
%
%
%
%
%
%

\section{Tolerant Local Testing}\label{sec:tol-test}

Recall that the tester of the Section~\ref{sec:ltall} was not a tolerant tester, i.e., it was not guaranteed to accept functions that are close to $\Fnd{n}{d}$ with high probability. Specifically a tolerant tester should given, $\delta_1 < \delta_2$ accept functions that are $\delta_1$-close with high probability while rejecting functions that are $\delta_2$-far with high probability. In this section, we augment the tester of Section~\ref{sec:ltall} to get such a tolerant tester. Theorem~\ref{thm:tolerant-Fq-intro}, gives the formal assertion. 


We start by describing our tolerant tester. At a high level, our tolerant tester estimates the distance of $f$ to the (unique) closest polynomial upto error $\epsilon = (\delta_2-\delta_1)/2$ with high probability. If the estimated distance $\mu$ is $< (\delta_1+\delta_2)/2$ it outputs YES, else it outputs NO.

\paragraph{Tester $T_{d,\delta_1,\delta_2}$}\label{tol_tester}
\begin{enumerate}
	\item Run the intolerant $(\frac{1}{2^{d+10}},2^{O(d)})$-tester of Theorem~\ref{thm:test-main} and reject if it rejects. (We assume that this rejects all functions that are $1/2^{d+10}$-far from $\Fnd{n}{d}$ with probability $ \geq 99/100$.)
	
	\item If the intolerant tester accepts,
	\begin{itemize}
		\item Choose a uniformly random map $\sigma: [n] \rightarrow [k]$ for $k = O(d/(\delta_2-\delta_1)^4)$.
		\item Choose an $a \in \{0,1\}^n$ uniformly at random. 
		\item Replace each $x_i$ with $y_{\sigma(i)} \oplus a_i.$
		\item Choose a set $S$ of $O(1/(\delta_2 - \delta_1)^2)$ points by sampling with replacement from $\{0,1\}^k$.	
	\end{itemize}
	
	\item 
	\begin{itemize}
		\item Interpolate the restricted function ${g(y_1, \ldots, y_k) = f(y_{\sigma(1)}\oplus a_1,\ldots, y_{\sigma(n)} \oplus a_n)}$ on $S$ and find the closest polynomial $h \in \Fnd{k}{d}$.
		\item If $\mu$ is the distance between $g$ and $h$ on the set $S$, i.e. 
		$$\mu = \delta_S(g,h) := \prob{x\in S}{g(x) \neq h(x)},$$ then accept if $\mu < (\delta_1+\delta_2)/2$ and reject otherwise.	
	\end{itemize}
	
\end{enumerate}

\bigskip

We will now present an analysis of the test to show that it has tolerant completeness and soundness guarantees. The first two lemmas are properties of the random map $(\sigma,a)$ and the set $S$, which when combined imply Theorem~\ref*{thm:tolerant-Fq-intro}. 

\begin{definition}
	A linear code $C = [N,K,D]_{\F}$ with generating matrix $G \in \F^{N \times K}$, is a linear subspace $C$ of dimension $K$ in $\F^N$, such that ${C = \{Gw\ |\  w \in \F^K\} }$. Moreover, the (absolute) distance between any two codewords is at least $D$, that is $D\leq N\cdot \min_{w,w' \in \F^K}\delta(Gw,Gw').$ (The latter quantity is also the minimum weight of any non-zero codeword.)
\end{definition}

\begin{lemma}\label{lem:random_S}
For a linear code $C = [N,K,D]_F$, a random set $S$ of size $O\left(\frac{NK\log N}{D} \right)$ is such that $C|_S$ (the code restricted to $S$) has fractional distance $O(D/N)$, with high probability. That is,
\begin{equation}
\prob{S \sim [N]^m}{\exists v_1 \neq v_2 \in C: \delta_S(v_1,v_2)  \leq \frac{D}{2N}} \leq 1/100.
\end{equation} 
\end{lemma}

\begin{proof}
For a vector $v \in \F^N$, we let $Z(v) = \{i \in [N]: v_i = 0\}$ and call this the zero set of $v$. For a linear subspace $C \subseteq \F^N$, let $\mathcal{Z}(C) = \{Z \subseteq [N]: \exists v \in C, Z = Z(v) \}$.
We will first show that  $|\mathcal{Z}(C)| \leq {N \choose \leq K}$. 

Let $w \in \F^K$ and $v = Gw$, where $G \in \F^{N \times K}$ is the generating matrix of the code $C$. For $i\in [N]$, let $r_i\in \F^K$ denote the $i$th row of the matrix $G$, and let $R(G) = \{r_1,\ldots,r_N\}$. The zero-set of each $v \in C$ corresponds naturally to a subset of $R(G)$, namely the set $R_v(G) := \{r_i \in G: i \in Z(v)\}$. Let $\mathcal{Z}(G) = \{R_v(G): v \in C\}$, then we have that $|\mathcal{Z}(C)| = |\mathcal{Z}(G)|$. We will now bound $|\mathcal{Z}(G)|$.

Notice that $R_v(G) = Ker(T_v: \F^K \rightarrow \F^K) \cap R(G)$,
where $T_w(x) = \langle w,x \rangle$ is the linear map on $\F^K$ defined by $w$. 
Hence if $x,y \in R_v(G)$ and $z = \alpha x+ \beta y \in R(G)$, then $z \in R_v(G)$. So, every set $R_v(G)$ has a one-one correspondence to a canonical basis of at most $K$ elements of $R(G)$, since any linear subspace of $\F^K$ can have dimension at most $K$. Therefore the total number of different sets in $\mathcal{Z}(G)$, can be at most ${N \choose \leq K}$, since $G$ has $N$ rows.
	
Now that we have bounded the size of $\mathcal{Z}(C)$ by ${N \choose \leq K}$, we would like to prove that the distance of the code restricted to a random set $S$ (chosen from $[N]$ without replacement) is preserved with high probability. Notice that the distance of the code $C|_S$ is the minimum weight of any codeword restricted to $S$. We can show this using a Chernoff and union bound argument as follows. Let $Z_v \subseteq [N]$ be a fixed zero-sets and $\overline{Z_v} = [N] \setminus Z_v$. Since the distance of the code $C$ is at least $D$, we know that $|\overline{Z_v}| \geq D$. A Chernoff bound gives us that,

\begin{equation}
\prob{S \sim [N]^m}{\frac{|\overline{Z_v} \cap S|}{|S|}  = \frac{\sum_i \mathbb{1}(s_i \in \overline{Z_v})}{|S|} \leq \frac{1}{2}\cdot  \frac{D}{N}} \leq \exp(-\frac{mD}{8N})
\end{equation}

We proved previously that the number of zero sets, $|\mathcal{Z}(v)|$ is at most ${N \choose K}$, hence we can union bound over these to get that, 

\begin{equation}
\prob{S \sim [N]^m}{\exists v \in C: \delta_S(v,\vec{0})  \leq \frac{D}{2N}} \leq {N \choose \leq K} \cdot \exp(-\frac{mD}{8N})
\end{equation}

Taking $m = O((N/D) \cdot K\log N)$, we get the statement of the lemma.

\end{proof}

%
%
%
%

\begin{corollary}\label{cor:prop_S}
Let $S \subseteq \{0,1\}^k$ be a randomly chosen set of size $m = k^{O(d)}$. Then we have that,
\begin{align*}
	\Pr_{S \sim (\{0,1\}^k)^m}[\exists p_1 \neq p_2 \in \Fnd{k}{d}, \delta_S(p_1,p_2) \leq 1/2^{d+1} ] \leq 1/100.
\end{align*}
\end{corollary}

\begin{proof}
We know that the space of degree $d$ polynomials, $\Fnd{k}{d}$ is a linear code, specifically, one can check that $\Fnd{k}{d} = [N,K,D]_{\F}$ with $N = 2^k, K = {k \choose d}$ and $D = 2^{k-d}$.	Applying Lemma~\ref{lem:random_S}, with $C = \Fnd{k}{d}$, we get the statement of the corollary.
\end{proof}

\begin{lemma}\label{lem:poly-prop}
	Let $f: \{0,1\}^n \rightarrow \F$ and $p \in \Fnd{n}{d}$ be $\delta$-close on $\{0,1\}^n$ with $0 \leq \delta \leq 1/2^{d+10}.$  Let $R_{\sigma,a}:\{0,1\}^k \rightarrow \{0,1\}^n$ be a random map defined by ${\sigma:[n] \rightarrow [k]}$ and $a \in \{0,1\}^n$ which are chosen uniformly at random and $R_{\sigma,a}(y_1,\ldots,y_k) = (y_{\sigma(1)} \oplus a_1,\ldots,y_{\sigma(n)} \oplus a_n)$.  Define the corresponding restrictions of $f,p$ as $g_{a,\sigma}(y) = f(R_{\sigma,a}(y))$ and $p_{a,\sigma}(y) = p(R_{\sigma,a}(y)), \forall y \in \{0,1\}^k$. Let $S \subseteq \{0,1\}^k$ be a set of size $m$, whose elements are chosen uniformly at random without replacement. For all $\epsilon \in [0,\delta]$, we have that, when $k = O(d\log(d/\epsilon)/\epsilon^4)$ and $m = O(1/\epsilon^2+k^d)$,
	
	\begin{equation} \label{eqn:prop1}
		\Pr_{a,\sigma,S}[\delta_S(p_{a,\sigma}, g_{a,\sigma}) \notin [\delta - \epsilon, \delta + \epsilon] ] \leq 2/100.
	\end{equation}
	\begin{equation}\label{eqn:prop2}
		\Pr_{a,\sigma,S}[\exists p' \neq p_{a,\sigma} \in \Fnd{k}{d}, \delta_S(p', g_{a,\sigma}) < \delta + \epsilon] \leq 2/100.	
	\end{equation}	
\end{lemma}

\begin{proof} Given two vectors $y,y' \in \{0,1\}^k$, let $\delta'(y,y')$ denote $\min(\delta(y,y') , 1 - \delta(y,y'))$, where $\delta(y,y')$ is the fractional Hamming distance between $y$ and $y'$. We will first prove that a randomly chosen set $S$ forms a good code, that is, with high probability, the points of $S \subseteq \{0,1\}^k$ are far from each other with respect to the measure of distance $\delta'$ defined above. Let us calculate the probability that a random set $S = \{s_1,\ldots,s_m\}$ of size $m$, has distance $\delta'(S) < \beta$ for some $\beta \in (0,1/2]$, where $\delta'(S) = \min_{i,j}(\delta'(s_i,s_j))$ ($m,k,\beta$ to be chosen later).  When the $m$ points of $S$ are picked, the probability that the $i^{th}$ point lies in the ball of radius $\beta$ around a fixed point $j < i$ is at most $2^{kH(\beta)}/2^k$, where $H(x)$ is the binary entropy function, so that the probability $\delta'(s_i,s_j) < \beta$ is at most $2 \cdot (2^{kH(\beta)}/2^k)$. So by a union bound we get that the probability that $\delta'(s_i,s_j)  < \beta$ for any $j < i$ is at most $(i-1)2^{kH(\beta) + 1}/2^k$. Now taking a union bound over all the $i$'s we get that the distance of the code $\delta'(S)  > \beta$ with probability at least $99/100$ if,

\begin{equation}\label{eqn:size_S}
m^2  \left(\frac{2^{k H(\beta) + 1}}{2^k} \right) \leq \frac{1}{100} ~~ \Longleftrightarrow ~~ k \geq  \frac{2\log m + \log 200}{1-H(\beta)}.
\end{equation}

From now on we will assume that $k,m$ are chosen appropriately so that $\delta'(S)$ is at least $\beta$ and prove the first part of Lemma~\ref{lem:poly-prop}, that equation~\ref{eqn:prop1} holds. Notice that each point in $\{0,1\}^k$ is mapped to a uniformly random point in $\{0,1\}^n$, via the random map $R = (\sigma,a)$, defined in the lemma. The points $s_i \in S$, generate the samples $R(S) = \{R(s_i)\} \subseteq \{0,1\}^n$. Let the error set of $f$ where it differs from $p$ be $T \subseteq \{0,1\}^n$ with density $\delta$. We will prove that the fraction of $T$ in $R(S)$ is a good estimate for the fraction of $T$ in $\{0,1\}^n$. We apply the Chebyshev inequality to the samples $R(S)$, to get that,

\begin{flalign}\label{eqn:tol-test}
\begin{split}
&\prob{R}{  \left \lvert \frac{\sum_i \mathbb{1}(R(s_i) \in T)}{m} - \avg{R}{ \frac{\sum_i \mathbb{1}(R(s_i) \in T)}{m} }\right \rvert  \geq \epsilon }  \\
&\leq \frac{\Var(\mathbb{1}(R(s_i) \in T ))}{m\epsilon^2} + \frac{m(m-1) (\Pr[ R(s_i), R(s_j) \in T] - \prob{}{R(s_i) \in T }\Pr[R(s_j) \in T]) } {m^2\epsilon^2}.	
\end{split}
\end{flalign}

\paragraph{}
All the samples $R(s_i)$ are uniformly distributed in $\{0,1\}^n$, so $\avg{R}{\mathbb{1}(R(s_i) \in T)} = \mu$. 
Since the distance between $s_i$ and $s_j$ is at least $\beta$ and at most $1-\beta$ (because $\delta'(S) < \beta$), the second sample $R(s_j)$ would be $\rho$-correlated to $R(s_i)$,\srikanth{Here, we probably only know that $\rho \leq 1-2\beta$? Do we have equality?} where $\rho \in [2\beta - 1, 1 - 2\beta]$ is the noise parameter of the noisy hypercube. By the small set expansion property of the noisy hypercube stated in Corollary~\ref{cor:noisy-hyp}, we have that,
$$\Pr[R(s_i),R(s_j) \in T]  \leq \mu^{\frac{2}{1+|\rho|}} \leq \mu^{\frac{1}{1-\beta}}.$$

Substituting these terms in~\ref{eqn:tol-test} we get,
\begin{equation}\label{eqn:chebyshev}
\prob{R}{  \left \lvert \frac{\sum_i \mathbb{1}(R(s_i) \in T )}{m} - \delta  \right \rvert  \geq \epsilon }   \leq \frac{1}{m \epsilon^2} + \frac{(\mu^{1/1 - \beta} - \mu^2)}{\epsilon^2}.
\end{equation}

This probability is less than $1/100$ for all $\mu$, if $m \geq O(1/\epsilon^2)$ and ${\beta = \frac{1}{2} - \frac{e\epsilon^2}{2}}$. We will choose $m = O(1/\epsilon^2) + O(k^d)$, $\beta$ as above, and $k \geq O(d\log(d/\epsilon)/\epsilon^4)$\srikanth{Probably pedantic b/c $\epsilon$ is quite small, but don't we need $k = d(\log (d/\epsilon))/\epsilon^4$ for the calculations?}\mitali{I've changed the expression of $k$ wherever it appears below. It does not affect the overall query complexity of the tolerant tester, so we don't need to make any further changes.} so that the parameters satisfy equation~\ref{eqn:size_S} (that $S$ forms a code with distance $\beta$) and equation~\ref{eqn:chebyshev} above.
One can check that ${\frac{\sum_i \mathbb{1}(R(s_i) \in T)}{m} = \delta_S(p_{a,\sigma},g_{a,\sigma})}$ and hence rephrasing equation~\ref{eqn:chebyshev} we get,

\begin{align*} 
&\Pr_{a,\sigma,S}[\delta_S(p_{a,\sigma}, g_{a,\sigma}) \notin [\delta - \epsilon, \delta + \epsilon] ] \\ \leq & \Pr_{a,\sigma,S}[S\text{ has distance} <\beta] + \Pr_{a,\sigma,S}[\delta_S(p_{a,\sigma}, g_{a,\sigma}) \notin [\delta - \epsilon, \delta + \epsilon] \mid S\text{ has distance } \geq \beta ] \\ \leq  &2/100.
\end{align*}

Now we will prove that the second property stated in the lemma, equation~\ref{eqn:prop2} holds. Lemma~\ref{cor:prop_S} gives us that with high probability, a random set $S$ is such that, all polynomials in $\Fnd{k}{d}$ have a small fraction of zeros, at most $1/2^{d+1}$ on $S$. Equivalently,
\begin{align}\label{eqn:good_S}
\Pr_{a,\sigma,S}[\exists p_1 \neq p_2 \in \Fnd{k}{d}, \delta_S[p_1, p_2] \leq 1/2^{d+1} ] \leq 1/100.
\end{align}

Equation~\ref{eqn:prop1} proved above, gives us that with high probability, $\delta_S(g_{a,\sigma},p_{a,\sigma}) \leq \delta + \epsilon \leq 1/2^{d+9}$. Additionally, when $S$ satisfies equation~\ref{eqn:good_S}, we know that, $\forall p' \neq p_{a,\sigma}, \delta_S(p',p_{a,\sigma}) \geq {1/2^{d+1} }$. Hence we get that with high probability (when both these events occur), $$\forall p' \neq p_{a,\sigma}, \delta_S(p',g_{a,\sigma}) \geq 1/2^{d+1} - 1/2^{d+9} \geq 1/2^{d+9} \geq \delta + \epsilon,$$ and hence equation~\ref{eqn:prop2} holds. 
\end{proof}


We will now show that these lemmas imply the main tolerant testing theorem.
\begin{proof}[Proof of Theorem~\ref{thm:tolerant-Fq-intro}]\label{proof-tol-test}
We will choose $\delta = 1/2^{O(d)}$ such that $\forall \delta_1,\delta_2 \leq \delta$, all functions that are $\delta_1$-close to $\Fnd{n}{d}$ are accepted by the tolerant tester~\ref{tol_tester} and all functions that are $\delta_2$-far are rejected, with high probability. 

\paragraph{Tolerant completeness:}
The intolerant tester of Theorem~\ref{thm:test-main} makes at most $2^{cd}$ queries, for some constant $c$. Each query is uniformly random in $\{0,1\}^n$, so the probability that the tester rejects a function that is $\delta$-close to $\Fnd{n}{d}$ is at most $2^{cd}\delta \leq 1/100$, if $\delta \leq 1/2^{O(d)}$. Let $f$ be $\delta_1 \leq \delta$-close to $\Fnd{n}{d}$ via the polynomial $p$. With high probability it won't get rejected in the first step of the tolerant tester.

In the rest of the tester, we choose a random map $R = (\sigma,a)$, a random set $S \subseteq \{0,1\}^k$ and then interpolate $g_{\sigma,a}$, the restricted function ($f|_{R(\{0,1\}^k)}$) on $S$. Let $\epsilon = (\delta_2 - \delta_1)/2$ and $p_{a,\sigma}$ be the restricted polynomial $p$. Then, Lemma~\ref{lem:poly-prop} gives us that, with high probability, $g_{\sigma,a}$ will be $\delta+\epsilon$-close to $p_{a,\sigma}$, when $k,m$ are chosen appropriately. Furthermore, all the other degree $d$ polynomials will be $\delta+\epsilon$-far from $g_{a,\sigma}$, so that $p_{a,\sigma}$ is the unique closest polynomial to which we decode. We know that $\delta_S(g_{a,\sigma,p_{a,\sigma}}) \leq \delta + \epsilon \leq (\delta_1 + \delta_2)/2$ and hence the tolerant tester accepts $f$.

\paragraph{Soundness:}
Let $f$ be $\delta_2$-far from $\Fnd{n}{d}$, we will show that the tolerant tester~\ref{tol_tester} rejects it. From Theorem~\ref{thm:test-main}, we get that with constant probability (at least $99/100$), the intolerant $(1/2^{d+10},2^{O(d)})$-tester will reject any function that is at least $1/2^{d+10}$-far from $\Fnd{n}{d}$. Assuming that the tester does not reject in the first step itself, we know that, the function $f$ is $1/2^{d+10}$-close to $\Fnd{n}{d}$ via the (unique) polynomial $p$. 

Let $p_{a,\sigma}, g_{a,\sigma}$ be the restricted polynomial $p$ and restricted function $f$ respectively. Similar to the completeness analysis given above, from Lemma~\ref{lem:poly-prop} we get that with high probability the tester decodes to the correct polynomial $p_{a,\sigma}$ and $\delta_S(p_{a,\sigma}, g_{a,\sigma}) \geq \delta + \epsilon \geq \delta_2$ and hence the tester rejects.

\paragraph{Query complexity and running time.}
To apply Lemma~\ref{lem:poly-prop}, we used $k = O(d\log(d/\epsilon)/\epsilon^4)$ and $m = O(1/\epsilon^2 + k^d)$, where $\epsilon = (\delta_2 - \delta_1)/2$. Our query complexity is thus $|S| = m = 2^{O(d\log (d/\epsilon^4))}$. (Note that this is significantly smaller than $2^{2^d}$). Our running time is however still doubly exponential in $d$ since the interpolation on $S$ in step(3) of the tester, might require exponential in $k$ running time since we do a brute-force interpolation. This is due to the random nature of the set $S$. (We note that a more careful construction of $S$ might yield the properties of Lemma~\ref{cor:prop_S} while providing enough structure to get a running time of $2^{O(d)}$. We leave this possibility as an open question.) 

\end{proof}

\bibliographystyle{plain}
\bibliography{polys-over-grids}

\begin{thebibliography}{10}

\bibitem{AKKLR}
Noga Alon, Tali Kaufman, Michael Krivelevich, Simon Litsyn, and Dana Ron.
\newblock Testing {Reed-Muller} codes.
\newblock {\em {IEEE} Trans. Information Theory}, 51(11):4032--4039, 2005.

\bibitem{ALMSS}
Sanjeev Arora, Carsten Lund, Rajeev Motwani, Madhu Sudan, and Mario Szegedy.
\newblock Proof verification and the hardness of approximation problems.
\newblock {\em J. {ACM}}, 45(3):501--555, 1998.

\bibitem{AroraS}
Sanjeev Arora and Shmuel Safra.
\newblock Probabilistic checking of proofs: {A} new characterization of {NP}.
\newblock {\em J. {ACM}}, 45(1):70--122, 1998.

\bibitem{BeaverF}
Donald Beaver and Joan Feigenbaum.
\newblock Hiding instances in multioracle queries.
\newblock In C.~Choffrut and T.~Lengauer, editors, {\em Proceedings of the 7th
  Annual Symposium on Theoretical Aspects of Computer Science}, pages 37--48,
  Rouen, France, 22--24 February 1990. Springer.
\newblock Lecture Notes in Computer Science, Volume 415.

\bibitem{BGW}
Michael Ben{-}Or, Shafi Goldwasser, and Avi Wigderson.
\newblock Completeness theorems for non-cryptographic fault-tolerant
  distributed computation (extended abstract).
\newblock In Janos Simon, editor, {\em Proceedings of the 20th Annual {ACM}
  Symposium on Theory of Computing, May 2-4, 1988, Chicago, Illinois, {USA}},
  pages 1--10. {ACM}, 1988.

\bibitem{BSHS}
Eli Ben{-}Sasson, Prahladh Harsha, and Sofya Raskhodnikova.
\newblock Some {3CNF} properties are hard to test.
\newblock {\em {SIAM} J. Comput.}, 35(1):1--21, 2005.

\bibitem{BKSSZ}
Arnab Bhattacharyya, Swastik Kopparty, Grant Schoenebeck, Madhu Sudan, and
  David Zuckerman.
\newblock Optimal testing of {Reed-Muller} codes.
\newblock In {\em 51th Annual {IEEE} Symposium on Foundations of Computer
  Science, {FOCS} 2010, October 23-26, 2010, Las Vegas, Nevada, {USA}}, pages
  488--497. {IEEE} Computer Society, 2010.

\bibitem{BCW}
Manuel Blum, Ashok~K. Chandra, and Mark~N. Wegman.
\newblock Equivalence of free {B}oolean graphs can be decided probabilistically
  in polynomial time.
\newblock {\em Inf. Process. Lett.}, 10(2):80--82, 1980.

\bibitem{BLR}
Manuel Blum, Michael Luby, and Ronitt Rubinfeld.
\newblock Self-testing/correcting with applications to numerical problems.
\newblock {\em Journal of Computer and System Sciences}, 47(3):549--595, 1993.

\bibitem{bonami}
Aline Bonami.
\newblock {\'E}tude des coefficients de fourier des fonctions de lp (g).
\newblock 20(2):335--402, 1970.

\bibitem{DeMilloL}
Richard~A. DeMillo and Richard~J. Lipton.
\newblock A probabilistic remark on algebraic program testing.
\newblock {\em Inf. Process. Lett.}, 7(4):193--195, 1978.

\bibitem{Dinur}
Irit Dinur.
\newblock The {PCP} theorem by gap amplification.
\newblock {\em J. {ACM}}, 54(3):12, 2007.

\bibitem{DP}
Devdatt Dubhashi and Alessandro Panconesi.
\newblock {\em Concentration of Measure for the Analysis of Randomized
  Algorithms}.
\newblock Cambridge University Press, New York, NY, USA, 1st edition, 2009.

\bibitem{JPRZ}
Charanjit~S. Jutla, Anindya~C. Patthak, Atri Rudra, and David Zuckerman.
\newblock Testing low-degree polynomials over prime fields.
\newblock {\em Random Struct. Algorithms}, 35(2):163--193, 2009.

\bibitem{KaufRon}
Tali Kaufman and Dana Ron.
\newblock Testing polynomials over general fields.
\newblock {\em {SIAM} J. Comput.}, 36(3):779--802, 2006.

\bibitem{KaufS}
Tali Kaufman and Madhu Sudan.
\newblock Algebraic property testing: the role of invariance.
\newblock In Cynthia Dwork, editor, {\em Proceedings of the 40th Annual {ACM}
  Symposium on Theory of Computing, Victoria, British Columbia, Canada, May
  17-20, 2008}, pages 403--412. {ACM}, 2008.

\bibitem{KimK}
John~Y. Kim and Swastik Kopparty.
\newblock Decoding {Reed-Muller} codes over product sets.
\newblock In Ran Raz, editor, {\em 31st Conference on Computational Complexity,
  {CCC} 2016, May 29 to June 1, 2016, Tokyo, Japan}, volume~50 of {\em LIPIcs},
  pages 11:1--11:28. Schloss Dagstuhl - Leibniz-Zentrum fuer Informatik, 2016.

\bibitem{Lipton}
Richard Lipton.
\newblock New directions in testing.
\newblock In {\em Distributed Computing and Cryptography}, volume~2 of {\em
  DIMACS Series in Discrete Mathematics and Theoretical Computer Science},
  pages 191--202. AMS, 1991.

\bibitem{LFKN}
Carsten Lund, Lance Fortnow, Howard~J. Karloff, and Noam Nisan.
\newblock Algebraic methods for interactive proof systems.
\newblock {\em J. {ACM}}, 39(4):859--868, 1992.

\bibitem{Muller}
D.~E. Muller.
\newblock Application of {B}oolean algebra to switching circuit design and to
  error detection.
\newblock {\em IEEE Transactions on Computers}, 3:6--12, 1954.

\bibitem{MVV}
Ketan Mulmuley, Umesh~V. Vazirani, and Vijay~V. Vazirani.
\newblock Matching is as easy as matrix inversion.
\newblock {\em Combinatorica}, 7(1):105--113, 1987.

\bibitem{donnell}
Ryan O'Donnell.
\newblock {\em Analysis of Boolean Functions}.
\newblock Cambridge University Press, 2014.

\bibitem{Polyanskiy}
Yury Polyanskiy.
\newblock Hypercontractivity of spherical averages in {Hamming} space.
\newblock {\em CoRR}, abs/1309.3014, 2013.

\bibitem{Razborov}
Alexander~A. Razborov.
\newblock On the method of approximations.
\newblock In David~S. Johnson, editor, {\em Proceedings of the 21st Annual
  {ACM} Symposium on Theory of Computing, May 14-17, 1989, Seattle, Washigton,
  {USA}}, pages 167--176. {ACM}, 1989.

\bibitem{Reed}
Irving~S. Reed.
\newblock A class of multiple-error-correcting codes and the decoding scheme.
\newblock {\em IEEE Transactions on Information Theory}, 4:38--49, 1954.

\bibitem{RubinfeldS}
Ronitt Rubinfeld and Madhu Sudan.
\newblock Robust characterizations of polynomials with applications to program
  testing.
\newblock {\em SIAM Journal on Computing}, 25(2):252--271, April 1996.

\bibitem{Schwartz}
Jacob~T. Schwartz.
\newblock Fast probabilistic algorithms for verification of polynomial
  identities.
\newblock {\em J. {ACM}}, 27(4):701--717, 1980.

\bibitem{Shamir79}
Adi Shamir.
\newblock How to share a secret.
\newblock {\em Commun. {ACM}}, 22(11):612--613, 1979.

\bibitem{Shamir90}
Adi Shamir.
\newblock {IP=PSPACE}.
\newblock In {\em 31st Annual Symposium on Foundations of Computer Science, St.
  Louis, Missouri, USA, October 22-24, 1990, Volume {I}}, pages 11--15. {IEEE}
  Computer Society, 1990.

\bibitem{Zippel}
Richard Zippel.
\newblock Probabilistic algorithms for sparse polynomials.
\newblock In Edward~W. Ng, editor, {\em Symbolic and Algebraic Computation,
  {EUROSAM} '79, An International Symposiumon Symbolic and Algebraic
  Computation, Marseille, France, June 1979, Proceedings}, volume~72 of {\em
  Lecture Notes in Computer Science}, pages 216--226. Springer, 1979.

\end{thebibliography}

\end{document}